\RequirePackage{fix-cm}
\documentclass[smallextended]{svjour3}       
\smartqed  
\usepackage{graphicx}
\usepackage{latexsym}
\usepackage{amssymb, amsmath, latexsym, amsfonts, mathrsfs, amsbsy}
\usepackage{epsfig, color}
\usepackage[percent]{overpic}
\usepackage{hyperref}

\newcommand{\revision}[1]{\textcolor{black}{#1}}

\newcommand{\R}{\mathbb R}
\newcommand{\neigh}{\mathcal N}
\newcommand{\leader}{\textsc L}
\newcommand{\follower}{\textsc F}

\begin{document}

\title{An all-leader agent-based model for turning and flocking birds
\thanks{This work was carried out within the research project ``SMARTOUR: Intelligent Platform for Tourism" (No. SCN\_00166) funded by the Italian Ministry of University and Research with the Regional Development Fund of European Union (PON Research and Competitiveness 2007–2013).
\\
The authors also acknowledge the Italian Ministry of Instruction, University and Research for supporting this research with funds coming from the project entitled Innovative numerical methods for evolutionary partial differential equations and applications (PRIN Project 2017, No. 2017KKJP4X). 
\\
EC is member of the INdAM Research group GNCS and MM is member of the INdAM Research group GNAMPA.
}
}

\author{Emiliano Cristiani \and Marta Menci \and Marco Papi \and L\'eonard Brafman}

\authorrunning{E. Cristiani, M. Menci, M. Papi, L. Brafman} 

\institute{E. Cristiani \at
           Istituto per le Applicazioni del Calcolo, Consiglio Nazionale delle Ricerche, Rome, Italy \\
           \email{e.cristiani@iac.cnr.it} 
\and
           M. Menci \at
              Istituto per le Applicazioni del Calcolo, Consiglio Nazionale delle Ricerche, Rome, Italy \\
           \email{m.menci@iac.cnr.it} 
\and
           M. Papi \at
              Universit\`a Campus Bio-Medico di Roma, Rome, Italy \\
           \email{M.Papi@unicampus.it} 
\and
           L. Brafman \at
           \email{leonard.brafman@gmail.com} 
}


\maketitle

\begin{abstract}
Starting from recent experimental observations of starlings and jackdaws, we propose a minimal agent-based mathematical model for bird flocks based on a system of second-order delayed stochastic differential equations with discontinuous (both in space and time) right-hand side. 
The model is specifically designed to reproduce self-organized spontaneous sudden changes of direction, not caused by external stimuli like predator's attacks. 
The main novelty of the model is that every bird is a potential turn initiator, thus leadership is formed in a group of indistinguishable agents. 
We investigate some theoretical properties of the model and we show the numerical results. Biolo\-gical insights are also discussed.  
\keywords{starlings \and turning \and agent-based models \and delay differential equations \and leaders \and switching leaders \and self-organization}
\subclass{92D50 \and 92B05}
\end{abstract}

\section{Introduction}\label{sec:intro}
Collective motion of bird flocks is an impressive spectacle of nature which always fascinated humans; see, e.g., \cite{sumpter2006,vicsek2012} for an introduction to the field. 
\revision{In recent years}, several experimental studies unveiled most of the underlying mechanisms -- mostly from a physics perspective -- by tracking and analyzing trajectories of flying birds. We refer in particular to the experiments performed in Rome (Italy) on starlings \cite{attanasi2014,attanasi2015,ballerini2008b,ballerini2008a,cavagna2010,cavagna2013,procaccini2011} and to other experiments made within similar scopes \cite{ling2019,pomeroy1992}.

Among the cited references, the two papers \cite{attanasi2015,ling2019} focus on the study of coherent changes in the direction of travel of the whole group, and represent the experimental foundations of our investigation.
More precisely, they consider collective turns that have a localized spatial origin, starting from a few individuals of the flock, and are triggered spontaneously, i.e.\ are not caused by external stimuli like a predator's  attack. 
Indeed, during aerial display, flocks of starlings often keep changing their direction of motion even in the absence of predators or obstacles, see Figure \ref{fig:realstarlings}.
\begin{figure}[h!]
    \centering
    \includegraphics[width=0.45\textwidth]{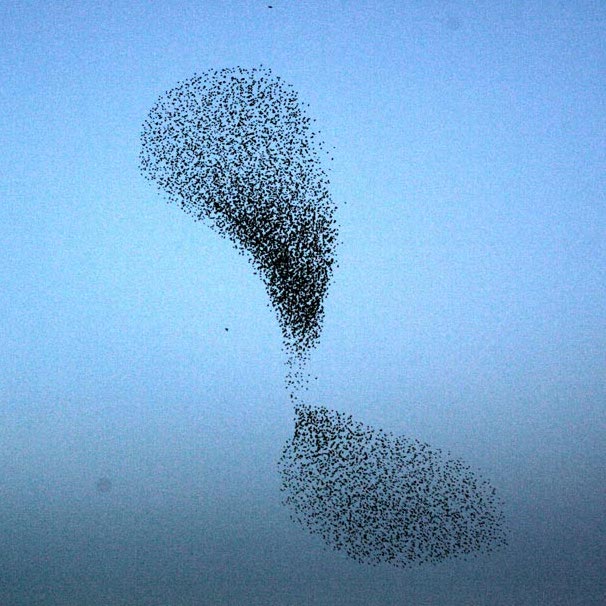}\hfill
    \includegraphics[width=0.524\textwidth]{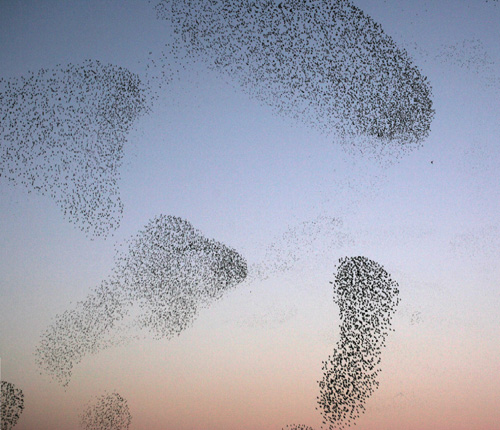}\\
    \caption{
    Two images of real starlings performing aerial displays. 
    The flock continuously changes direction of motion, stretches, compresses, and occasionally splits and then reunites.
%
\revision{
Courtesy: 
COBBS Lab, Institute for Complex Systems, National Research Council of Italy, Rome.
}
    }
    \label{fig:realstarlings}
\end{figure}
When a turn initiator triggers a change of direction, the information spreads along the flock as a wave, reaching rapidly all the group mates \cite{attanasi2015}. Nevertheless, in very large flocks, two or more changes of direction could take place at the same time. In fact, turning initiators can be so far from each other that turning waves propagate independently within the flock for a certain time. In that case, \revision{sometimes} the flock splits. 

\revision{Another important biological mechanism which serves as a foundation our work is that of \textit{switching leaders}, which has been observed in several animal groups, including flocks of birds \cite{butail2019detecting,chen2016switching,mwaffo2018detecting}. 
In large groups undergoing complex maneuvers, the group might experience an instantaneous change of its heading direction, resulting in a temporary change of group leadership.}

\medskip

Beside experimental literature, a number of numerical simulations were proposed to reproduce collective animal behavior. 
Seminal papers \cite{couzin2005,CS2007,vicsek1995} have shown that simple algorithms based on basic rules can generate an apparently complex self-organized (leaderless) collective motion. 
\revision{Paper \cite{bernardi2021} investigates, in general terms, how leaders can steer a group of animals.}

In this work, we put ourselves in the context of three-zone agent-based models, see, e.g., \cite{couzin2002,eftimie2018,sumpter2006,vicsek2012}, which are widely accepted in the community.
In particular, the simulator \textit{Stardisplay} \cite{hemelrijk2019,hemelrijk2011,hildenbrandt2010} aims at reproducing the aerial display of starlings both from a qualitative and quantitative point of view. That simulator is quite sophisticated (it has as many as 23 parameters), including also a simplified aerodynamics. Nevertheless, the model is not able to reproduce spontaneous turns not caused by roost attraction or by external stimuli.

\revision{Particularly worth mentioning is also} the inertial spin model \cite{caglioti2020,cavagna2015}. 
The model, although quite effective in reproducing the turning dynamics, does not consider the trigger of the turns, which is instead the main goal of this paper.

Other important features which are included in our model are already present in the literature: we mention in particular, \textit{time-delayed communications} among agents \cite{dong2020,haskovec2020,hemelrijk2011,hildenbrandt2010} and the well known \textit{topological interactions} \cite{ballerini2008a,dong2020,hemelrijk2011,hildenbrandt2010}, i.e.\ the fact that agents interact with a fixed number of group mates, regardless of their distance.

\paragraph{Paper contribution.}
In this paper we present a new agent-based mathematical model for bird flocks based on a system of second-order delayed stochastic differential equations with discontinuous (both in space and time) 
right-hand side, specifically designed to reproduce self-organized sudden changes of direction. We focus on spontaneous turns, i.e.\ not caused by roost attraction or predators' attacks, and on very large flocks.
We think that the proposed model is \textit{minimal}, in the sense that all the terms constituting the model are actually necessary for reproducing the turning effect. In particular,  topological interactions appear to be essential in order to reproduce desired turning behavior.

The model is based on two seemingly contradictory ideas: 
1) Turning effect is a genuine self-organized phenomenon arising in bird flocks, meaning that the new directions of motion arise from local interactions between group members. 
There is no permanent special member which lead the flock, all birds being indistinguishable.
2) Among birds, a sort of leadership actually exists, in the sense that one can find birds behaving differently from the others. 
The apparent contradiction is overcome by the fact that \textit{all birds occasionally try to change direction}, then all birds equally act as group controllers. 
If they are followed by others, they keep the new direction, otherwise they cease moving solo and return to the group. 
In this way indistinguishability is preserved because all birds in the flock have the same role.
Moreover, no bird is indispensable, meaning that removing any part of the flock does not affect the behavior of the remaining birds. 

From the analytical point of view, we show that the initial value problem associated to the model is well-posed. 
Starting from the classical approach in \cite{filippov1988book} (see also \cite{cristiani2011}), we extend well-established results for systems of differential equations with discontinuous right-hand side, in order to take into account the presence of the delay and the discontinuity with respect to time.

\paragraph{Paper organization.} 
In Section \ref{sec:model} we present the mathematical model.
In Section \ref{sec:theory} we study the theoretical properties of the model. 
In Section \ref{sec:numericalresults} we present the numerical results and finally 
in Section \ref{sec:conclusions} we discuss the results and sketch some conclusions.

\section{The model}\label{sec:model}
Let us consider a group of $n>1$ agents represented by dimensionless points having unit mass, moving in the three-dimensional space. We assume they are labeled univocally by their index $k=1,\ldots,n$. 
Let us denote by $X_k(t)=(X^1_k(t),X^2_k(t),X^3_k(t))\in\R^3$ and $V_k(t)=(V^1_k(t),V^2_k(t),V^3_k(t))\in\R^3$ the position and the velocity of the $k$-th agent at time \revision{$t \ge 0$}, respectively. 

\subsection{Leaders, followers, and their dynamics}
Agents can be either \textit{leaders} or \textit{followers}: the status of the agent $k$ at time $t$ is given by the function
\begin{equation}\label{status}
s_k(t)=:
\left\{
\begin{array}{ll}
\leader, & \ \ \mbox{if $k$ is a leader at time $t$,}\\
\follower, & \ \ \mbox{if $k$ is a follower at time $t$,}
\end{array}
\right.,\qquad k=1,\ldots,n.
\end{equation}
In our model, the role of the leaders is to \textit{initiate turns}, deviating from the common direction of motion.
In accordance to the typical biological findings, we assume that agents are indistinguishable, i.e.\ we reject the idea of the existence of a hierarchical structure within the flock. 
All agents are initially followers, then they have equal probability to become leaders. 
The change of status follower$\to$leader is ruled by a stochastic process, more precisely the switching time is a random variable with exponential distribution. In a discrete setting, considering a final time $T$ for the dynamics, we simply adopt a geometric distribution where a decision (keeping the follower status or becoming a leader) is taken every time step $\Delta t$ of the numerical simulation.
On the contrary, the change of status leader$\to$follower is deterministic and it is ruled by two model parameters: the \textit{persistence time} $\mathfrak p$ and the \textit{persistence distance} $\mathfrak d$. 
If either an agent has been leader for over $\mathfrak p$ time units or the distance from its nearest neighbor is over $\mathfrak d$ space units, then it returns to be a follower. 
Indeed, it was observed that turns are triggered by birds which deviate from the group direction and keep the new direction for a while \cite{attanasi2015} (see also \cite{toulet2015} in the context of sheep). 

It is also biologically sound that when an agent comes back to the follower status, the probability of becoming a leader again immediately thereafter is very low (cf.\ \cite[Box 2]{couzin2009} in the context of ants). 
Our model accounts for this fact introducing a third parameter, called \textit{\revision{refractory} time} $\mathfrak r$, which corresponds to the period of time a follower cannot change its status after being a leader. 

\medskip

The dynamics of the agents are described by the following system of delayed ordinary differential equations
\begin{equation}\label{eq:model}
\left\{
	\begin{array}{l}
	\dot X_k(t)=V_k(t),\\ [2mm]
	\dot V_k(t,s_k)=A_k(t-\delta,s_k(t)):=
	\left\{
	\begin{array}{ll}
	A_k^{\leader}(t-\delta), & \text{ if } s_k(t)=\leader, \\ [2mm]
	A_k^{\follower}(t-\delta), & \text{ if } s_k(t)=\follower, 
	\end{array}
	\right.
    \end{array}
\right.
\end{equation}
for $k=1,\ldots,n$ and $t\geq \delta$, where $\delta>0$ is the time delay (reaction time).
$A_k$ is the force field the agent $k$ is subject to, which also depends on its status. 

\subsection{Interactions}\label{subsec_interactions}

We consider a classical three-zone model, where repulsion, alignment and attraction forces among agents coexist. 
Repulsion is needed to avoid collisions, alignment to get flocking behavior, and attraction for group cohesion.

Following \cite{ballerini2008a}, we consider pure topological interactions, meaning that the dynamics of each agent $k$ are influenced by a fixed number $M$ of nearest mates only, regardless of their distance.
To do that, at any time $t$, and for each agent $k$, we order all the agents from the closest to the farthest w.r.t.\ the $k$-th one, i.e.\ we get the index set $\{1_k(t),\ldots,n_k(t)\}$ such that $\|X_{1_k(t)}-X_k(t)\|\leq \cdots \leq \|X_{n_k(t)}-X_k(t)\|$.
We denote by 
\begin{equation}\label{eq:defNeigh}
    \neigh(k,t;M):=\{1_k(t),\ldots,M_k(t)\}
\end{equation}
the set of the $M$ nearest neighbors of agent $k$ at time $t$, see Figure \ref{fig:neighbors}.
\begin{figure}[h!]
	\centering
	\begin{overpic}[width=0.95\textwidth]{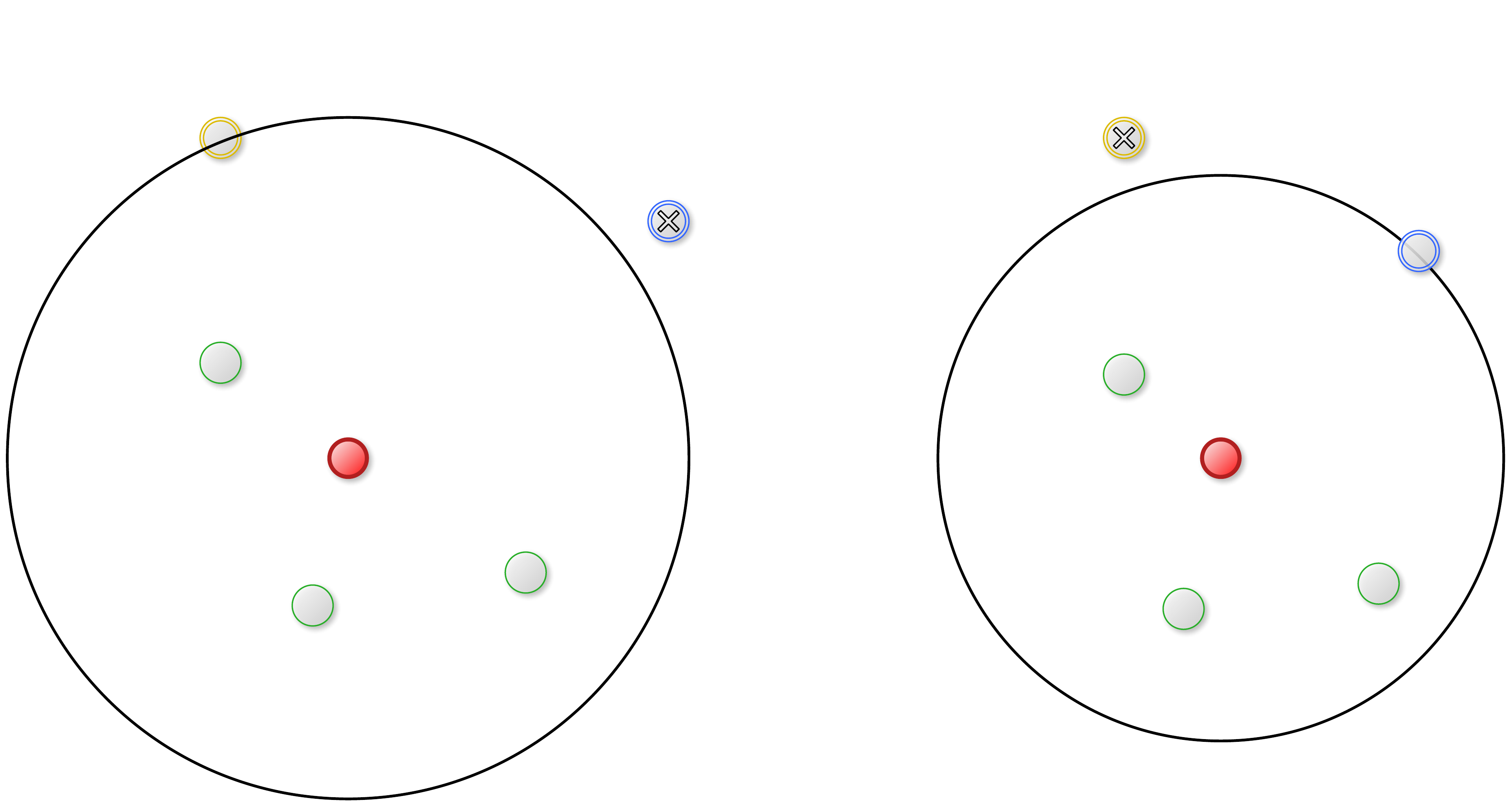}
	\put (-5,0) {\textbf{a.}} \put (60,0) {\textbf{b.}}
	\put (21,20) {$k$}  \put (79,20) {$k$}
	\put (12,40) {$M_k$} \put (89,34) {$M_k$}
	\put (43,34) {($M$+1)$_k$} \put (77,44) {($M$+1)$_k$}
	\put (-3,26) {\large $\mathcal N$} \put (58,26) {\large $\mathcal N$}
\end{overpic}
	\caption{Neighbors of agent $k$ with $M=4$. From \textbf{a.} to \textbf {b.} we observe the switch between the $M$-th and the ($M$+1)-th neighbor, with consequent redefinition of the set of neighbors $\mathcal N$ and the threshold distance $\|X_{M_k}-x_k\|$.}
	\label{fig:neighbors}
\end{figure}
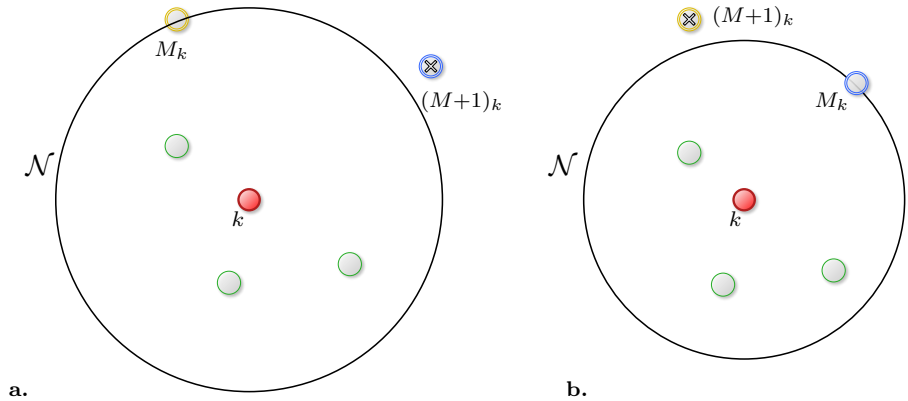
\begin{remark}\label{rem:Neighambiguous} 
The definition of $\neigh(k,t;M)$ is ambiguous in the case the $M$-th and the ($M$+1)-th nearest neighbor are at the same distance, i.e.\ $\|X_{M_k}-X_k\|=\|X_{(M+1)_k}-X_k\|$.
To avoid any ambiguity and/or stochasticity in the definition of $\mathcal N$, in the case of two equidistant neighbor we select the agent with the lower index.
\end{remark}

We define the following social forces:
\begin{enumerate}
\item Repulsion

\begin{equation}\label{eq:repulsionforce}
A_k^{rep}(t):= -C^{rep}\sum_{j_k \in \neigh(k,t;M)} \frac{X_{j_k}(t)-X_k(t)}{\|X_{j_k}(t)-X_k(t)\|^2 +\varepsilon} 
\end{equation} 
\item Alignment 
\begin{equation}\label{eq:alignmentforce}
A_k^{ali}(t):=\displaystyle\frac{C^{ali}}{M}\displaystyle\sum_{j_k \in \neigh(k,t;M)} (V_{j_k}(t)-V_k(t))
\end{equation}
\item Attraction 
\begin{equation}\label{eq:attractionforce}
A_k^{att}(t):= C^{att}\sum_{j_k \in \neigh(k,t;M)} (X_{j_k}(t)-X_k(t))
\end{equation}
with $\varepsilon$, $C^{rep}$, $C^{ali}$, $C^{att}$ positive parameters. 
\end{enumerate}
 
We observe that attraction and repulsion forces are function of the distances among the agents. In particular, attraction force grows with the distance, whereas the repulsion force decreases. 
Parameter $\varepsilon$ \revision{is a small positive constant which} avoids degenerate repulsion forces and, from the modeling point of view, translates the fact that agents are not really dimensionless.

 
Finally, we set the dynamics of the agent $k$ as in \eqref{eq:model} with 
\begin{equation}\label{eq:completedynamics}
     A_k^{L}:=A_k^{rep}
     , \qquad 
     A_k^{F}:=A_k^{rep}+A_k^{ali}+A_k^{att}.
\end{equation}
It is quite natural assuming that repulsion force is exerted to every agent, either leader or follower, since all agents need to avoid collisions with other group members. 
Instead, only followers display cohesion and alignment attitudes, whereas leaders, willing to steer the flock, naturally leave the flock under the repulsion force.

We want to stress once again that in our model each agent can be a leader, regardless of its position within the flock. In accordance to the experimental literature \cite{attanasi2015}, we expect that all the turning attempts performed by leaders in the interior of the flock likely fail, since the strong repulsion forces coming from all directions oblige the leader to keep the relative position inside the flock.
Leaders at the boundary of the flock, instead, experience non-symmetric interactions with group mates, then they could likely succeed in steering the flock since the resulting repulsion force points outside the flock. 
This is especially true if, by chance, two or more agents at the same time become leader, they are close to each other, and their repulsion forces point approximately in the same direction outside the flock. 
Under these conditions, a sufficient critical turning force is created and it makes leaders' neighbors follow leaders, generating the desired cascade effect.  



\section{Analytical results}\label{sec:theory}
In this section we show that the initial value problem associated to system \eqref{eq:model} is well-posed. 
For simplicity, we assume that the switching times for the follower$\rightarrow$leader change of status are sampled \textit{a priori} and therefore can be assumed to be known in advance.
This reinterpretation drops the stochasticity of the model, which now becomes fully deterministic.

We first state a general result concerning existence and uniqueness of the solution to delay differential equations with discontinuous rhs. This result (Proposition \ref{prop:Filippov_extension} in the following) can be regarded as an extension of classical Filippov's results, see \cite[Chapter 1, Theorems 1-2]{filippov1988book}.
After that, we show that the specific model presented in Section \ref{sec:model} falls in this general theory, thus proving the well-posedness of the problem (Theorem \ref{teo_exuniq}).


\medskip

Let us consider the following system of differential equations with constant delay $\delta \in (0,T)$, 
$f: [0,T] \times \Omega \times \Omega \rightarrow \mathbb{R}^n$ with $\Omega \subset \mathbb{R}^n$ open
\begin{align}\label{model_delay}
\left\{
	\begin{array}{ll}
	\dot{y}(t)=f\left(t, y(t), y(t-\delta) \right), & t \in (\delta,\delta+T), \\ [2mm]
	y(s)= \phi_{0}	\in \Omega, &  s \in [0,\delta].
	\end{array}
\right.
\end{align}

We assume $f$ satisfies the \textit{Caratheodory conditions}:

\begin{itemize}
\item[A1)] The function $f\left(t, y, w \right)$ is continuous in $(y,w)$ for almost all $t$;
\item[A2)] The function $f\left(t, y, w \right)$ is measurable in $t$ for each $(y,w)$;
\item[A3)] $\left| f\left(t, y, w \right) \right| \le m(t)$, being $m(t) $ a summable function.
\end{itemize}

Moreover, we assume that
\begin{itemize}
\item[A4)] there exists a summable function $l(t)$ such that for any points $(t,y_1,w_1)$ and $(t,y_2,w_2)$ $\in [0,T] \times \Omega \times \Omega$ it holds
\begin{equation}\label{per_unicita}
\left| f\left(t, y_1, w_1 \right) - f\left(t, y_2, w_2 \right) \right| \le l(t) \left( \left|  y_1- y_2 \right| +  \left|  w_1- w_2 \right|\right).
\end{equation}
\end{itemize}
\begin{proposition}\label{prop:Filippov_extension}
Under assumptions A1)-A4), there exists at most one solution to the problem (\ref{model_delay}) in  $(\delta, \delta + T)$. 
\end{proposition}
The proof of Proposition \ref{prop:Filippov_extension} is postponed in the Appendix.

\medskip 

Let us now introduce some notations which will be used in the following.
We define the norm of the whole system configuration $X=[X^\top_1,\ldots,X^\top_n]\in \mathbb{R}^{3\times n}$ as $||X||=\left[\sum_{k=1}^n ||X_k||^2\right]^{1/2}$
\revision{and we denote by $B_{r}(\zeta)$ the ball centered at $\zeta$ with radius $r>0$.}
%
We note that the set $\mathcal N(k,t;M)$ depends on $t$ only via the system state $X(t)$. Therefore, with a small abuse of notation, we will also refer to this set by $\mathcal N(k,X;M)$.


%
%

We denote by
%
\begin{align}\label{def:S}
\mathcal S:=\left\{
X\in \mathbb{R}^{3\times n}\ \text{ s.t.\ } ||X_{M_k}-X_k||= || X_{(M+1)_k} -X_k||\ 
\right. 
\nonumber \\ 
\left.
\text{ for some } k\in\{1,\ldots,n\}
\right\}
\end{align}
the set of \textit{switching configurations}, corresponding to the situation where some agent $k$ is equally distant from its $M$-th and ($M$+1)-th nearest neighbor. 
Note that the complement set $\mathcal S^c:=\mathbb{R}^{3\times n}\backslash \mathcal S$ is open and dense in $\mathbb{R}^{3\times n}$. The following remark is crucial for the rest of the theoretical investigation.
\begin{remark}\label{rem:Emi'scontribution} 
If a system configuration $X(t)\in\mathcal S^c$ for an interval of time $\Delta$, then the set $\mathcal N(k,X;M)$ of indices of the $M$ nearest neighbors of the agent $k$ remains identical for any $t\in\Delta$.
In fact, the set of indices of the $M$ nearest neighbors changes only if the $M$-th and the ($M$+1)-th nearest neighbors swap (while the ordering of the $M$ nearest neighbors could change, but this does not affect the set of indices). 
Most important, if the set $\mathcal N(k,X;M)$ does not change for any agent $k$, then no agent changes the flock mates it is interacting with, see \eqref{eq:repulsionforce}-\eqref{eq:alignmentforce}-\eqref{eq:attractionforce} and Figure \ref{fig:neighbors}.
\end{remark}



In order to apply Proposition \ref{prop:Filippov_extension}, it is useful to rewrite our dynamical system in a compact form. 
We denote by $Z=[Z_1,\ldots,Z_n]\in \mathbb{R}^{3\times 2\times n}$ the variable whose $k$-th component is $Z_k=[X^\top_k,V^\top_k]\in \mathbb{R}^{3 \times 2}$. 
Without loss of generality, in this section we assume $L=1$ and $F=0$ in (\ref{status}), and using (\ref{eq:completedynamics}), we rewrite $A_k$ as
\begin{eqnarray}\label{mp.12}
A_k(Z,s_k):=-C^{rep} \sum_{j_k\in\mathcal N (k,X;M)} r(X_{j_k}-X_k)+ \nonumber\\
+(1-s_k) \sum_{j_k\in\mathcal N (k,X;M)} \left[\frac{C^{ali}}{M}(V_{j_k}-V_k)+C^{att} (X_{j_k}-X_k)\right],
\end{eqnarray}
where $r:\mathbb{R}^3\rightarrow \mathbb{R}^3$, $r(\zeta):=\frac{\zeta}{\|\zeta\|^2+\varepsilon}$.

Thus, the model in (\ref{eq:model}) is equivalent to the following delayed system of differential equations: 
\begin{eqnarray}\label{mp.11}
\dot{Z}_k(t)=H_k(Z(t-\delta),s_k(t)),\qquad t\geq \delta,
\end{eqnarray}
where $$H_k(Z,s_k):=[V^\top_k,A^\top_k(Z,s_k)].$$ 


We consider the initial value problem associated to (\ref{mp.11}) with constant initial data given by $X_k(t)=X_k^0$, $V_k(t)= V_k^0$, for any $t\in [0,\delta]$.
 
\medskip 
 
From the analytical point of view, one of the main issue of our model is represented by the leader activation function $s_k(\cdot)$, which implies a discontinuous right-hand side in (\ref{mp.11}) with respect to $t$.
For that reason, we need to specify a suitable definition for a solution to our problem. 
To this end, we extend the notion of Filippov in \cite[Ch.\ 1]{filippov1988book}.
 

\begin{definition} 
A function $Z=Z(t)$ is a local solution to (\ref{mp.11}) satisfying the initial condition $Z^0\in \mathbb{R}^{3\times 2\times n}$ if there exists $T>0$ such that $Z$ is absolutely continuous on each closed interval $I\subset [\delta, \delta+T)$ and satisfies 
\begin{eqnarray}\label{mp.13}
Z_k(t)&=&Z_k^0+\int_\delta^t H_k(Z(\tau-\delta),s_k(\tau))\,d\tau,\quad \forall\;\; \revision{\delta}< t<\delta+T,\\
Z_k(t)&=&Z_k^0,\qquad \forall\;\; 0\leq t\leq \delta,
\end{eqnarray}
for any $k=1,\ldots,n$.
\end{definition}

We now state our main result, concerning the existence and uniqueness of a solution in the sense of the definition given above.

\begin{theorem}\label{teo_exuniq}
For every $X^0=(X_1^0,\ldots,X_n^0)\in {\mathcal S^c}$, $V^0=(V_1^0,\ldots,V_n^0)\in \mathbb{R}^{3\times n}$ the delayed initial value problem (\ref{mp.11}) with initial datum $Z^0$, with components $Z^0_k=[(X_k^0)^\top,(V_k^0)^\top]$, $k=1,\ldots,n$, admits at most one solution. 
\end{theorem}

\begin{proof}
In order to prove the result, we show that $(Z,t)\mapsto H_k(Z,s_k(t))$ satisfies the assumptions of Proposition \ref{prop:Filippov_extension}. To this end, it suffices to study the properties of $(Z,t)\mapsto A_k(Z,s_k(t))$. 
Since $X^0\in {\mathcal S^c}$, there exists $\nu_0>0$ such that $\mathcal{N}(k,X;M)=
\mathcal{N}(k,X^\prime;M)=
\mathcal{N}(k,X^0;M)=:
\bar{\mathcal N}(k,\nu_0;M)$ 
for any $X$, $X^\prime$ such that $||X-X^0||<\nu_0$ and  $||X^\prime-X^0||<\nu_0$, see Remark \ref{rem:Emi'scontribution}. 
%
Let us denote by $\ell$ the generic index belonging to $\bar{\mathcal N}(k,\nu_0;M)$.
Therefore, for any $k$ we get the following inequalities:
\begin{eqnarray}\label{mp.14}
||(X_\ell-X_k)-(X^\prime_\ell-X^\prime_k)||\leq  ||X_\ell-X^\prime_\ell||+||X_k-X^\prime_k||
\end{eqnarray}
and, similarly, for any $V, V^\prime\in \mathbb{R}^{3\times n}$,
\begin{eqnarray}\label{mp.15}
||(V_\ell-V_k)-(V^\prime_\ell-V^\prime_k)||\leq  ||V_\ell-V^\prime_\ell||+||V_k-V^\prime_k||.
\end{eqnarray}

For each $\mu_0>0$ we consider the following open domain:
\begin{eqnarray}\label{mp.16}
\Omega^0 :=\!\!\left\{Z=[X^\top,V^\top]\in \mathbb{R}^{3\times 2\times n}\;\Big{|}\; \mbox{ $\begin{array}{c} 
X=(X_1,\ldots,X_n)\in B_{\nu_0}(X^0)\\
\!\!\!\!V=(V_1,\ldots,V_n)\,\,\in B_{\mu_0}(0)\ \ 
\end{array}$}\right\}.
\end{eqnarray}
Thus, we show that the Caratheodory conditions A1)-A3) are satisfied by $A_k(Z,s_k(t))$, in the open domain $\Omega^0$.

Since $r$ is continuous and satisfies $||r(\zeta)-r(\zeta^\prime)||\leq ||\zeta-\zeta^\prime||/\varepsilon$, for all $\zeta$, $\zeta^\prime$ in $\mathbb{R}^3$, by (\ref{mp.14})-(\ref{mp.15}), we get the following inequality:
\begin{eqnarray}\label{mp.17}
&&||A_k(Z,s_k)-A_k(Z^\prime,s_k)||\leq \frac{2M}{\varepsilon} C^{rep} || X-X^\prime ||\nonumber\\
&&+2|1-s_k| C^{ali} ||V-V^\prime||+2M |1-s_k| C^{att} ||X-X^\prime||\nonumber\\
&&\leq \left[
\frac{2M }{\varepsilon}C^{rep}+2 \left(C^{ali} +MC^{att}\right) |1-s_k|\right]||Z-Z^\prime||,
\end{eqnarray}
for any couple $Z$, $Z^\prime$ $\in \Omega^0$. 
Hence $A_k(Z,s_k(t))$ is continuous in $Z\in \Omega^0$, for all $t$. Since $s_k\mapsto A_k(Z,s_k)$ is a linear function, it is clearly continuous in $s_k$, for each $Z$, implying that $t\mapsto A_k(Z,s_k(t))$ is measurable in $t$, for each $Z$. 
The Caratheodory conditions A1)-A2) are satisfied.

Moreover, by the definition of $A_k$ in (\ref{mp.12}) and taking the maximum of the function $r$, we obtain
\begin{equation}\label{mp.18}
||A_k(Z,s)||\leq \frac{C^{rep} M}{2\sqrt{\varepsilon}} + 2|1-s|(C^{ali}\mu_0+M C^{att}\nu_0)
\end{equation}
for any $Z\in \Omega^0$, $s\in \mathbb{R}$. Since $s_k(t)\in \{0,1\}$, for each $t$, we deduce that the condition A3) holds true. 
Moreover, by inequality (\ref{mp.17}), 
we deduce that $(t,Z)\mapsto A_k(Z,s_k(t))$ satisfies also condition A4), with
$l(t):=2M C^{rep}/\varepsilon+2 (C^{ali} +MC^{att}) (1-s_k(t))$ in (\ref{per_unicita}). 
Hence, Proposition (\ref{prop:Filippov_extension}) yields the existence and the uniqueness of a local solution. 
\qed
\end{proof}
\noindent
\begin{remark} We observe that the solution to the problem (\ref{mp.13}) is defined in a suitable interval of time $[\delta,\delta+T)$, where $T>0$ may depend on the constants $\nu_0$, $\mu_0$.
\end{remark}

\section{Numerical results}\label{sec:numericalresults} 
In this section we present the results of some simulations carried out by means of the model \eqref{eq:model}-\eqref{eq:completedynamics}.
Contrary to the simulator \textit{Stardisplay} \cite{hemelrijk2019,hemelrijk2011,hildenbrandt2010}, we do not calibrate parameters on the basis of measured values. We prefer to focus on qualitative, rather than quantitative, description of the flock behavior, with special attention to the turning behavior. 
Parameters are normalized with respect to the characteristic length of a bird and a fictitious time unit, both set to 1. 
Parameters used in the simulations are: 
$M=7$, $C^{rep}=2.5$, $C^{ali}=3$, $C^{att}=0.01$, $\delta=0.1$, $\frak p=700$, $\frak d=20$, $\frak r=800$.
                                                                                                                                                
At initial time $t=0$, agents are \revision {uniformly} random distributed in a cube (a square in the case of the 2D test) of side \revision{200}, with null velocity.
For the numerical approximation, we use a standard explicit Euler scheme with time step $\Delta t=0.1$.
The transition follower$\to$leader is possible every time step and the probability of the transition (Bernoulli trial) is $2\times 10^{-4}$.

%
%
%

\subsection{A sample 2D flock}
We start our numerical analysis from a sample 2D test with 200 agents. 
\revision{We would like to stress that the reduction to a bi-dimensional test is particularly convenient for presenting the main features of the model, and in particular the effect of the leaders. Nevertheless, our final goal is to obtain the most realistic results possible in the 3D case, then model parameters are chosen to this end.}

Figure \ref{test:baricentro2D} shows the trajectory of the barycenter of the flock. We compare the case of the dynamics without (\revision{Figure \ref{test:baricentro2D}\textbf a}) and with (\revision{Figure \ref{test:baricentro2D}\textbf b}) leaders, \revision{starting with the same initial condition}. 
In the former case no agent becomes leader, then no turn is initiated. 
The flock starts moving under the three main social forces \eqref{eq:repulsionforce}-\eqref{eq:alignmentforce}-\eqref{eq:attractionforce} only and, after a transient, the flock finds a consensus, i.e.\ all agents reach the same velocity (magnitude and direction). 
On the contrary, if a sufficient number of leaders are present, the flock cannot reach the consensus because the equilibrium is continuously broken by the changes of direction. The barycenter's trajectory is less regular and never stabilizes, and several turnings are observable.
\begin{figure}[h!]
    \centering
    \textbf a.\includegraphics[width=0.47\textwidth]{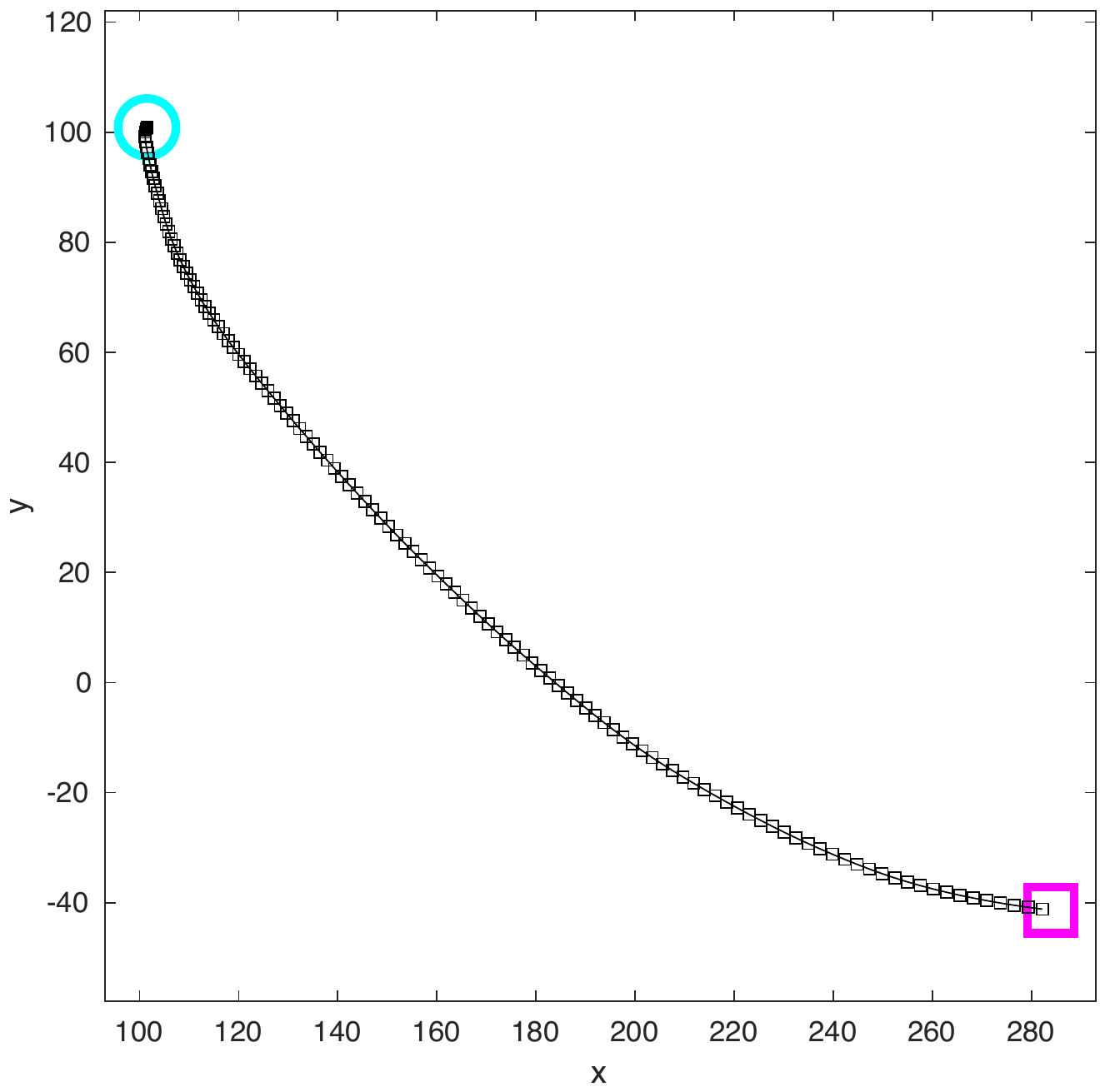}\hfill
    \textbf b.\includegraphics[width=0.47\textwidth]{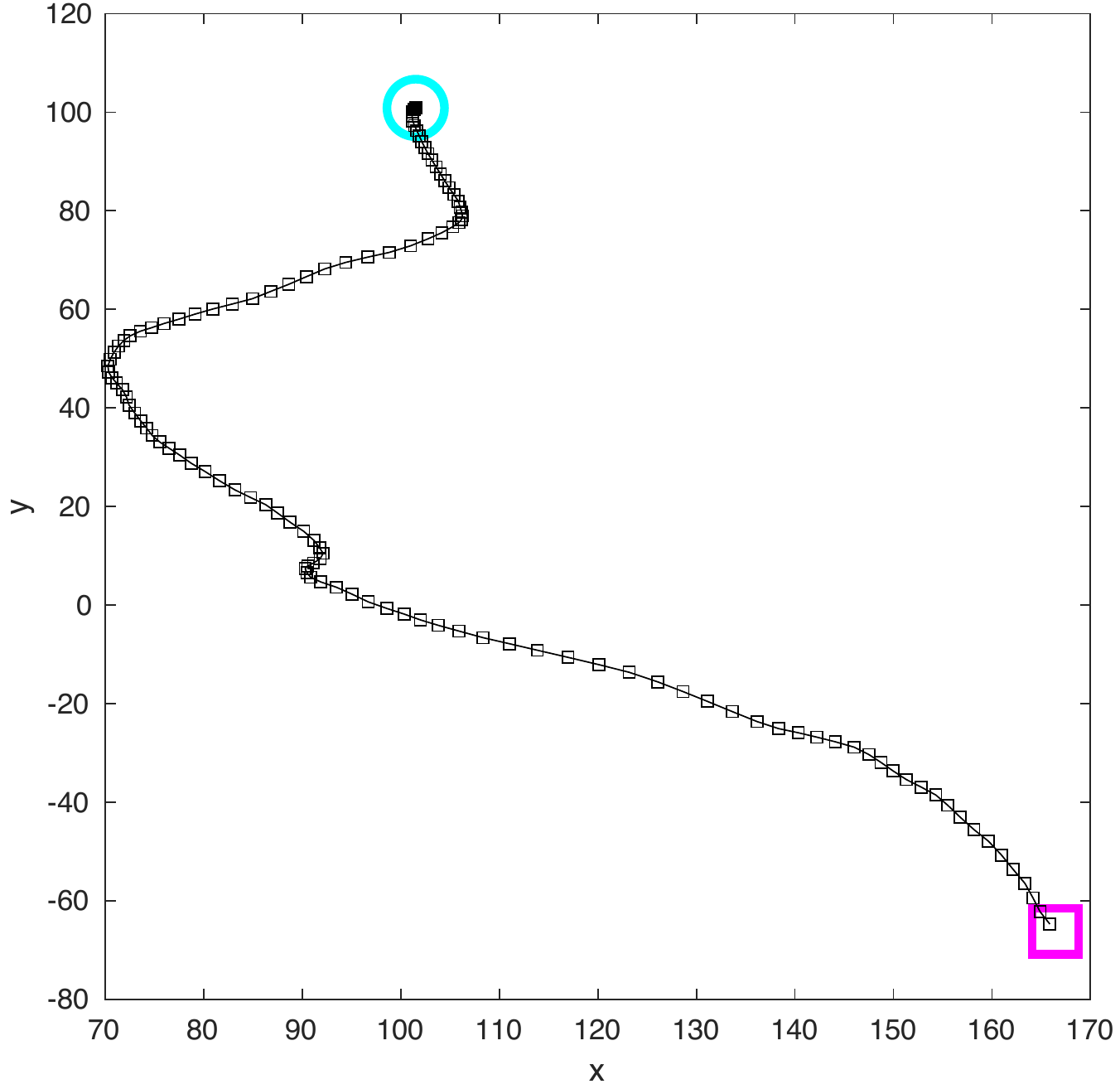}
    \caption{Test 2D: trajectory of the barycenter of the flock without (\textbf a.) and with (\textbf b.) leaders. \revision{Cyan circle and magenta square denote, respectively, the initial and final positions of flock barycenter.}}
    \label{test:baricentro2D} 
\end{figure}

\newpage
\revision{Figure \ref{test:mean_std} shows the time evolution of the velocities of the group members, in terms of both magnitude and direction. 
More precisely we plot, at each time $t$, the mean and the standard deviation of the norms of the velocities of the agents $\{ \|V_k(t)\| \}_{k=1,\ldots,n}$
for the scenario without (Figure \ref{test:mean_std}\textbf a) and with leaders (Figure \ref{test:mean_std}\textbf b).
Similarly, Figure \ref{test:mean_std}\textbf c and \ref{test:mean_std}\textbf d show, for each time $t$, the mean and the standard deviation of the angles $\{\theta_k(t)\}_{k=1,\ldots,n}$, 
where $\theta_k(t)$ is defined as the angle between the horizontal direction and $V_k(t)$.}

\revision{As already observed, if no agents become leaders, the flock reaches a consensus state, hence a common velocity (see Figure \ref{test:mean_std}\textbf a for magnitude and Figure \ref{test:mean_std}\textbf c for direction).
%
On the contrary, the presence of leaders involves frequent and sudden changes, both in the magnitude (Figure \ref{test:mean_std}\textbf b) and in the direction (Figure \ref{test:mean_std}\textbf d) of the velocities.
}
\begin{figure}[h!]
    \centering
    \textbf a.\includegraphics[width=0.47\textwidth]{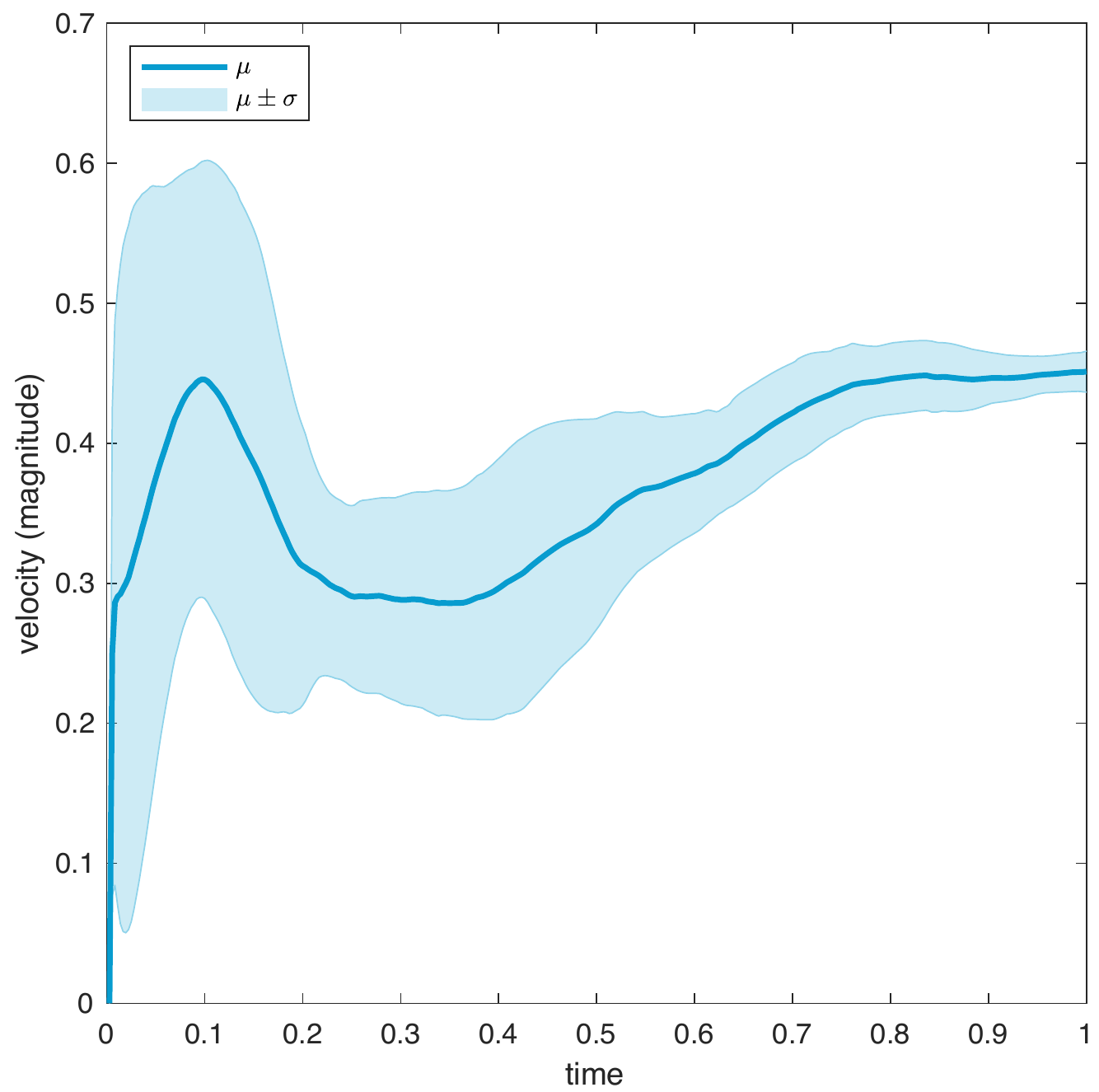}\hfill
    \textbf b.\includegraphics[width=0.47\textwidth]{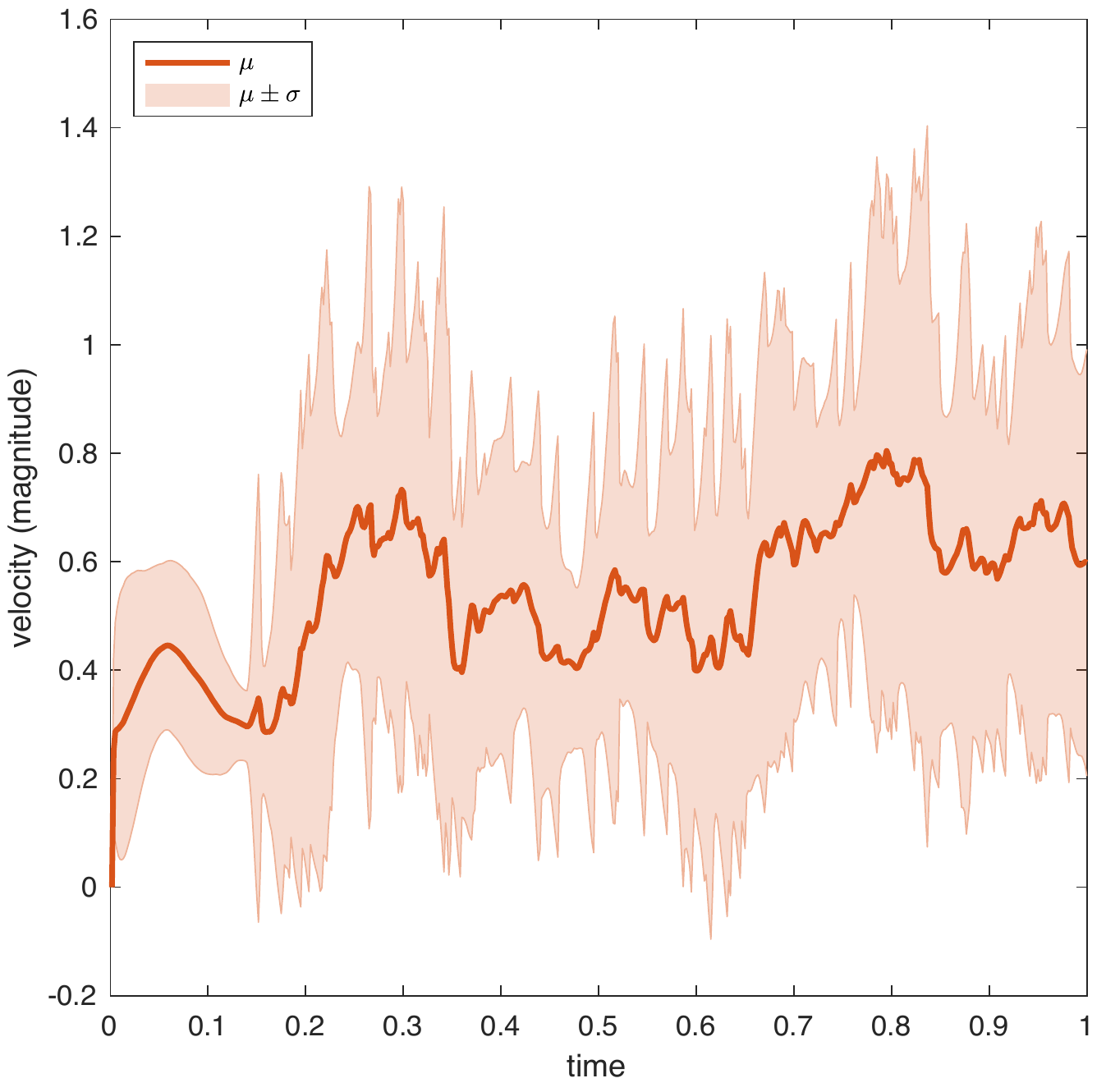}\\
    \textbf c.\includegraphics[width=0.47\textwidth]{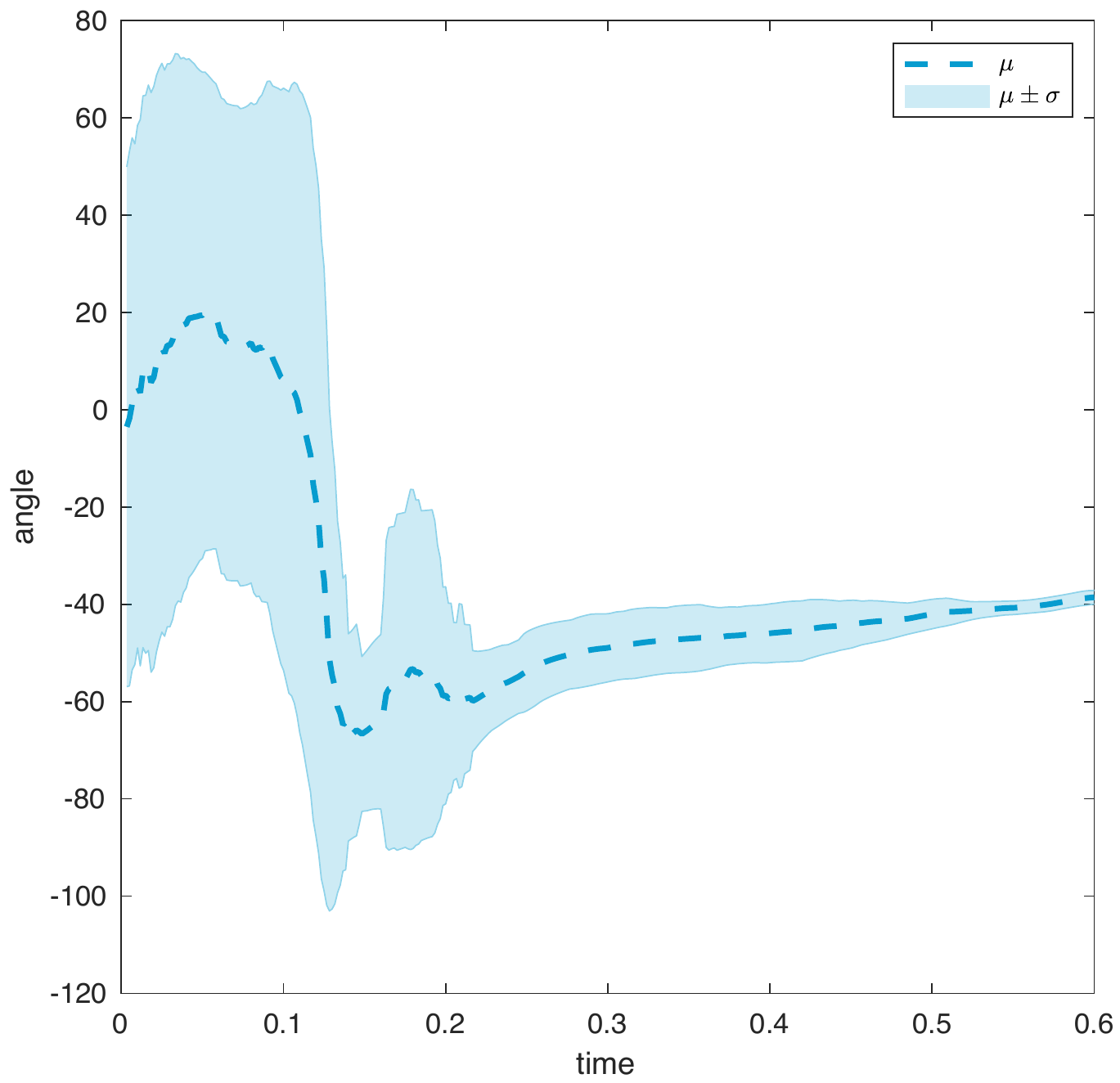}\hfill
    \textbf d.\includegraphics[width=0.47\textwidth]{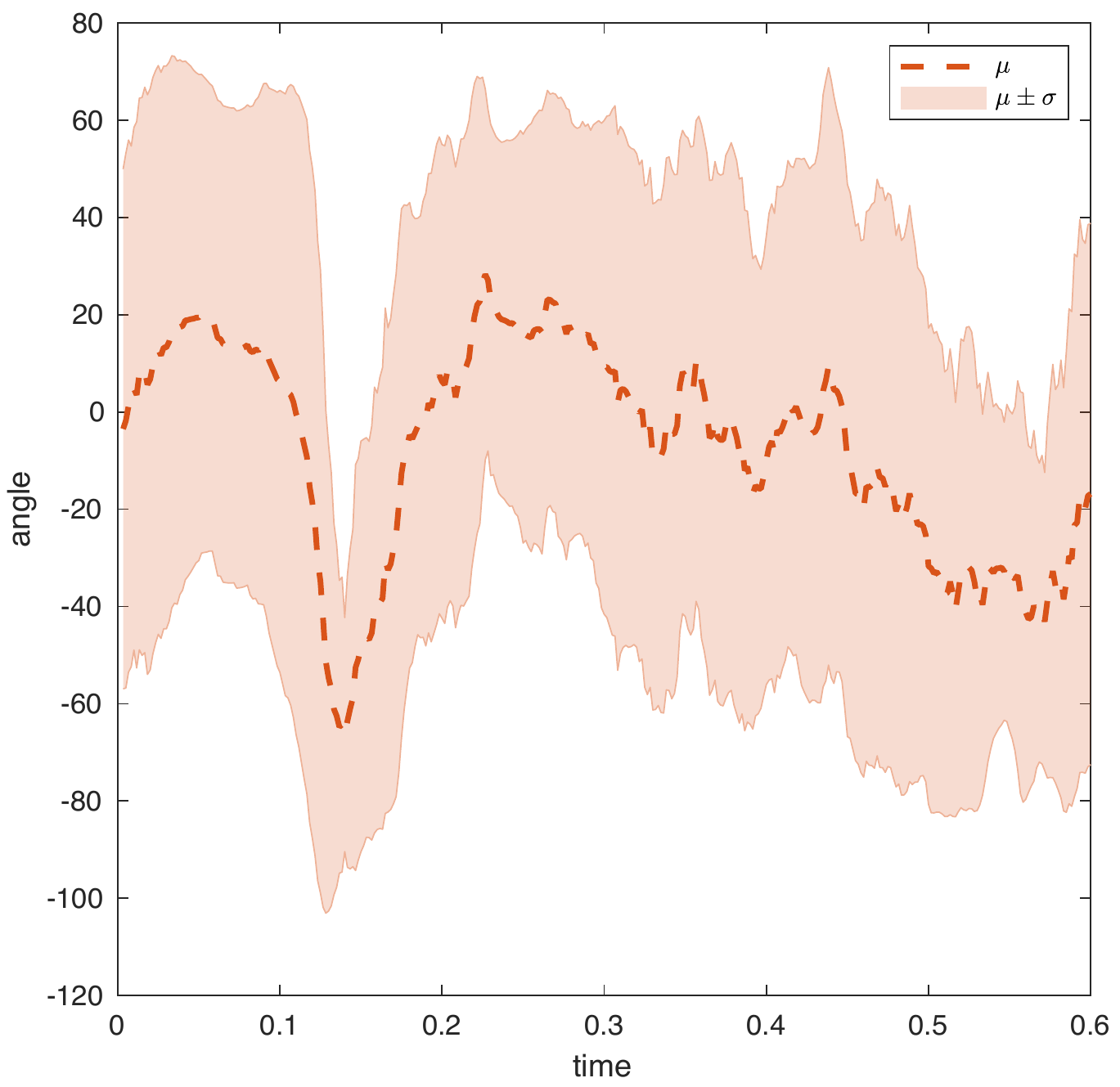}
    \caption{\revision{Test 2D:  Plot of mean and standard deviation of the modulus of the velocities (\textbf a. \& \textbf b.) and the angles between the velocities and the horizontal direction (\textbf c. \& \textbf d.), in the case without leaders (\textbf a. \& \textbf c. ) and with leaders (\textbf b. \& \textbf d.). 
    The value of the angles is expressed in degrees.}}
    \label{test:mean_std} 
\end{figure}

\clearpage
Figure \ref{test:elongazione2D} shows the horizontal and vertical elongations of the flock, respectively defined by
$$
e_h(t):=\max_k\{X^1_k(t)\}-\min_k\{X^1_k(t)\},\qquad
e_v(t):=\max_k\{X^2_k(t)\}-\min_k\{X^2_k(t)\}.
$$
As before, we compare the case of the dynamics without (\revision{ Figure \ref{test:elongazione2D}\textbf a}) and with (\revision{ Figure \ref{test:elongazione2D}\textbf b}) leaders. It is evident that in presence of leaders the flock continuously stretches and elongates in all directions.
\begin{figure}[h!]
    \centering
    \textbf a.\includegraphics[width=0.47\textwidth]{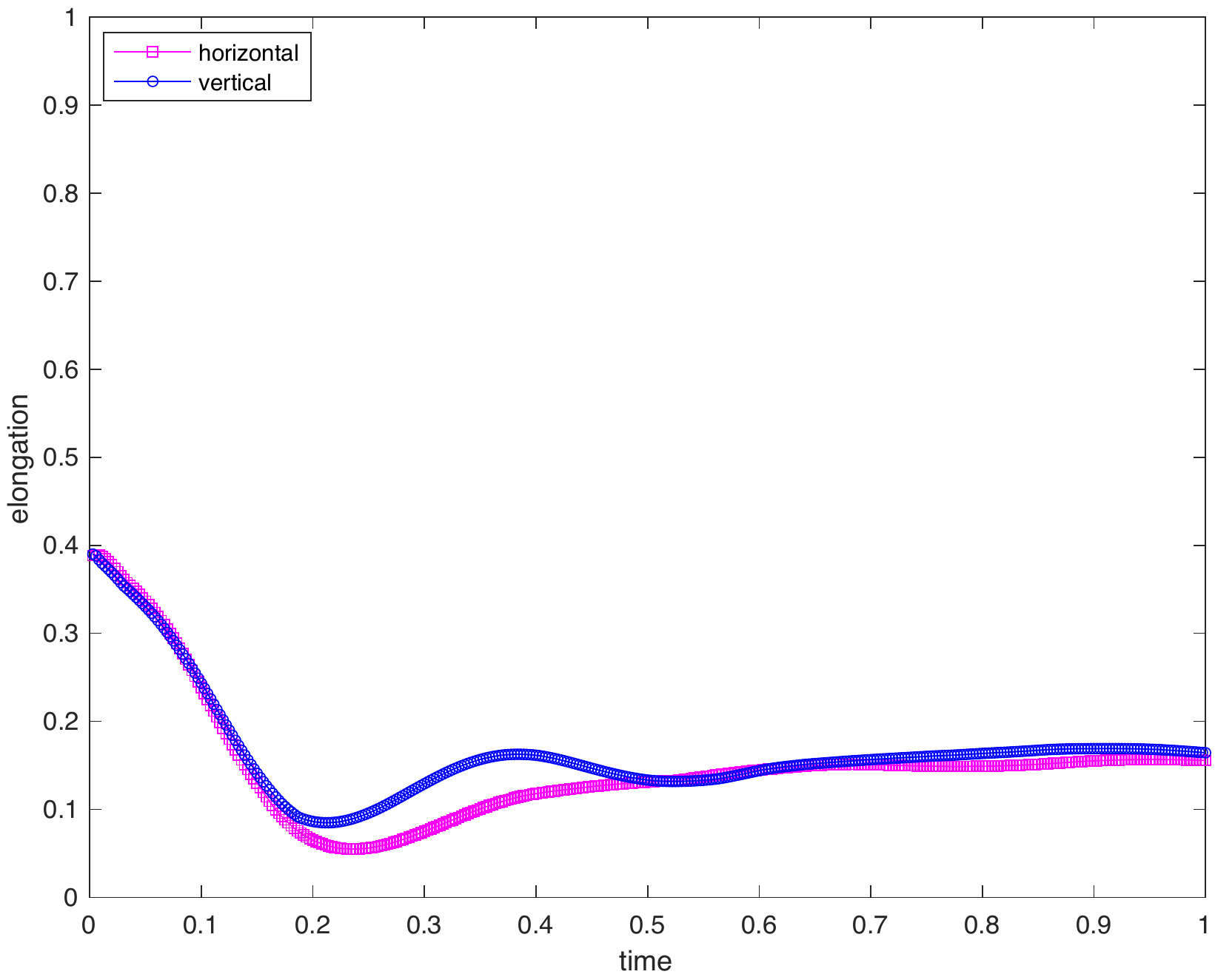}\hfill
    \textbf b.\includegraphics[width=0.47\textwidth]{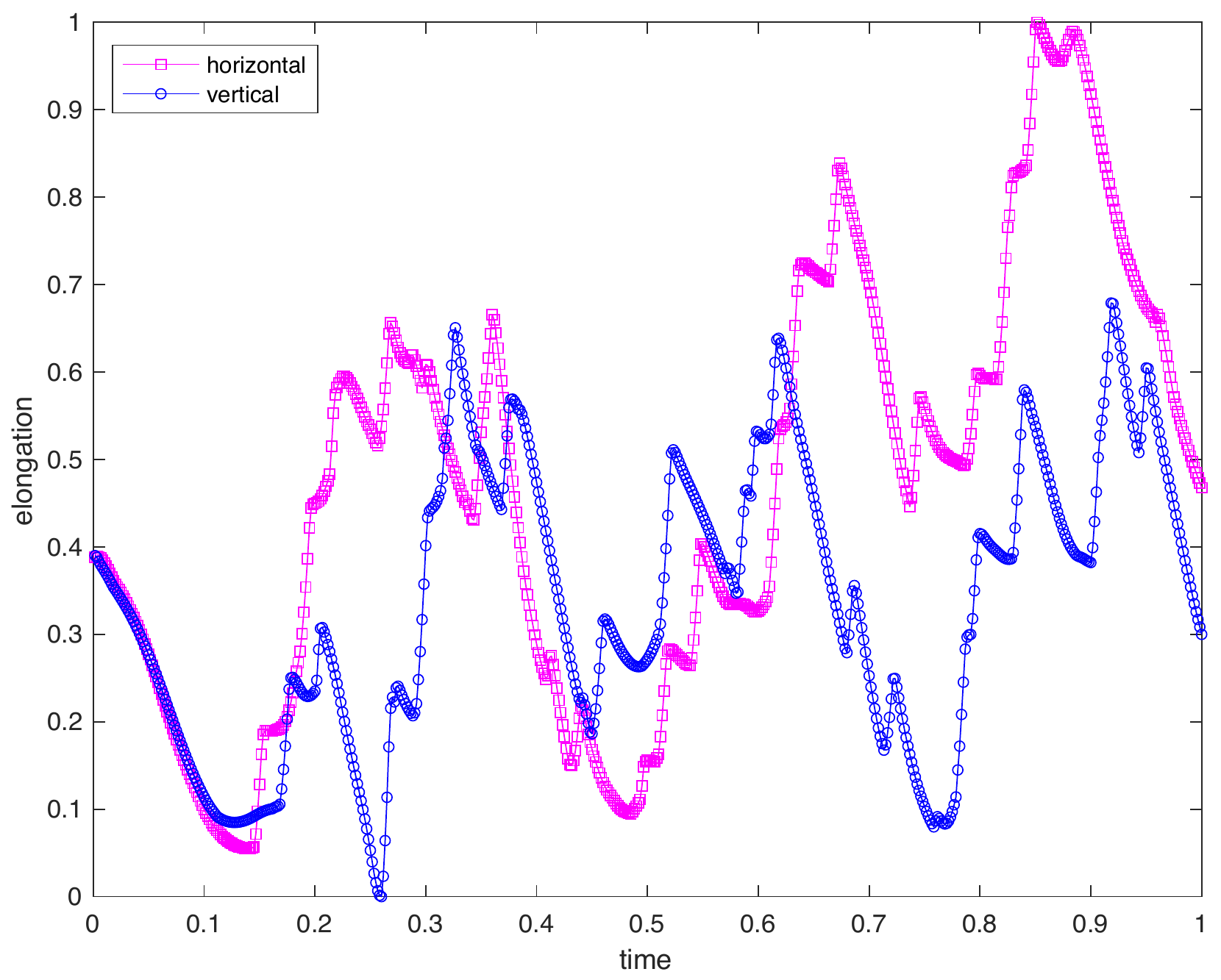}
    \caption{Test 2D: horizontal and vertical elongations of the flock without (\textbf a.) and with (\textbf b.) leaders. Elongations are measured in percentage with respect to the minimal and maximal elongations reached during simulation \revision{with leaders}.}
    \label{test:elongazione2D} 
\end{figure}
%

\newpage
\revision{Figure \ref{test:new_test2D} shows two screenshots of the numerical simulations performed without and with leaders.
Until the first follower$\to$leader transition (Figure \ref{test:new_test2D}\textbf b), the flock behaves as in the case without leaders (Figure \ref{test:new_test2D}\textbf a).
The dynamics begin to differ as leaders appear, confirming the results about the flock elongation: without leaders, the flock finds a consensus (Figure \ref{test:new_test2D}\textbf c). 
Conversely, with leaders the flock continuously stretches in different directions (Figure \ref{test:new_test2D}\textbf d). 
Figure \ref{test:new_test2D}\textbf d corresponds to the flock configuration at the maximal horizontal elongation (see Figure \ref{test:elongazione2D}\textbf b). }
\begin{figure}[h!]
    \centering
    \textbf a.\includegraphics[width=0.47\textwidth]{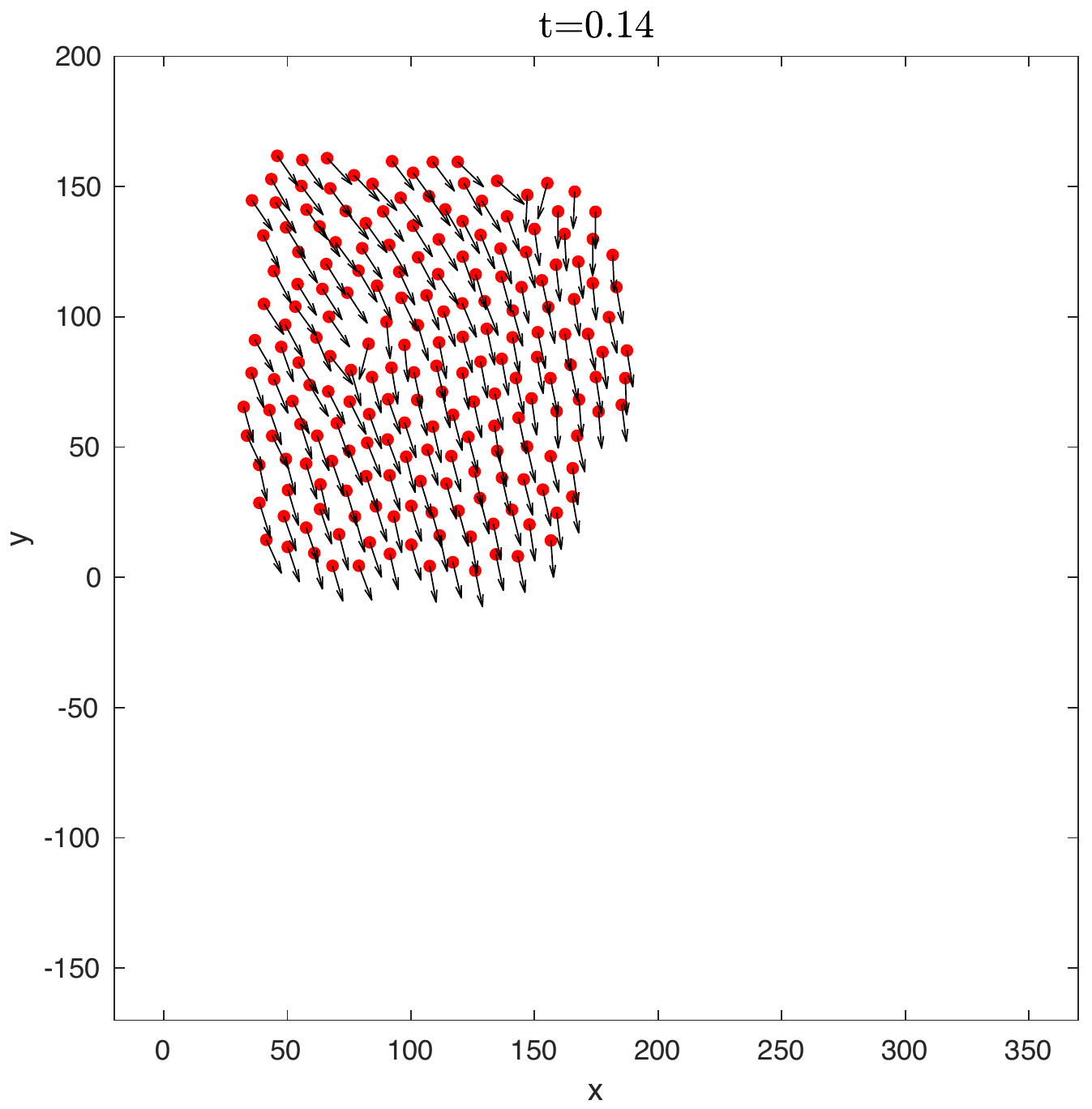}\hfill
    \textbf b.\includegraphics[width=0.47\textwidth]{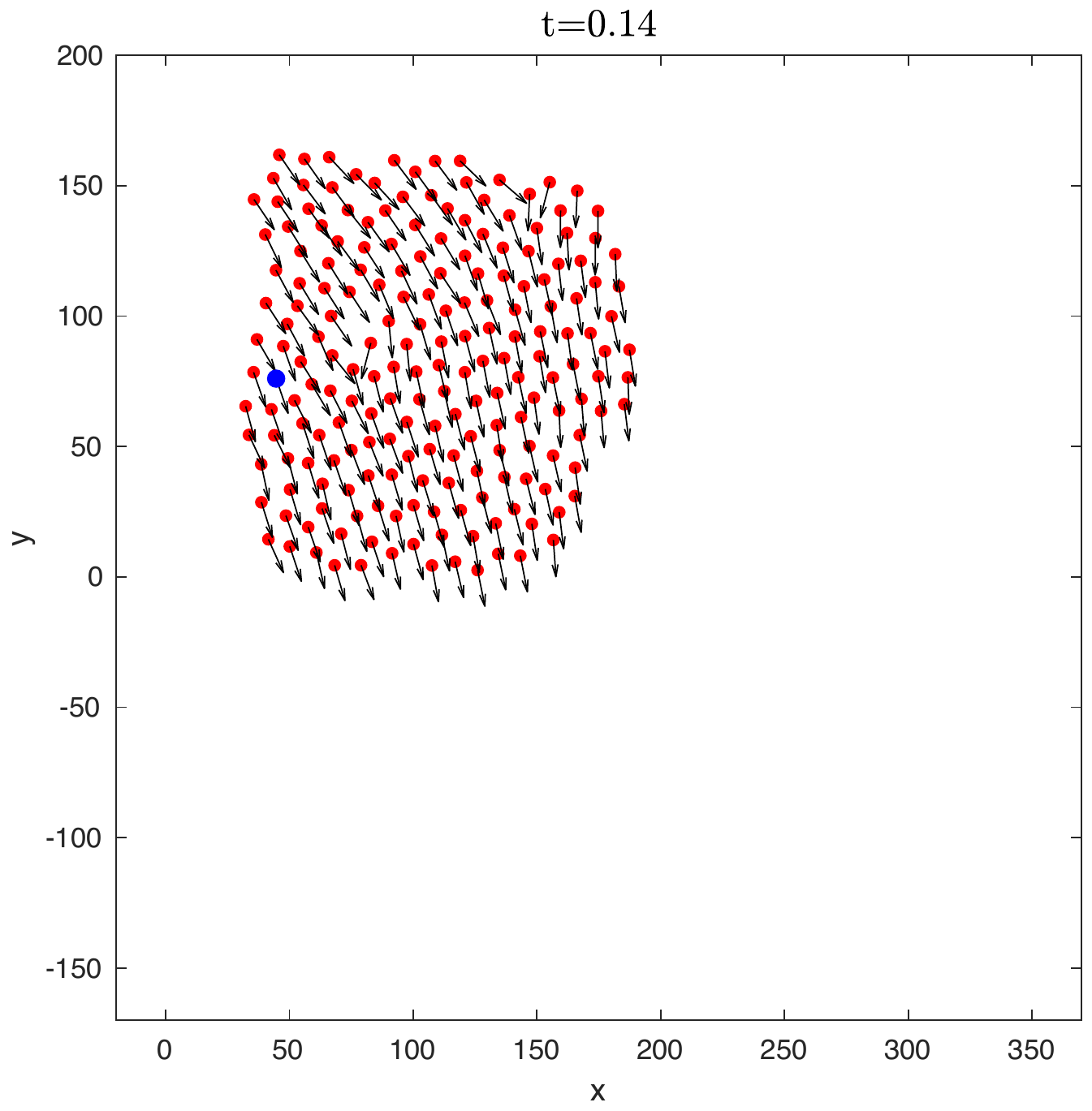}
    \textbf c.\includegraphics[width=0.47\textwidth]{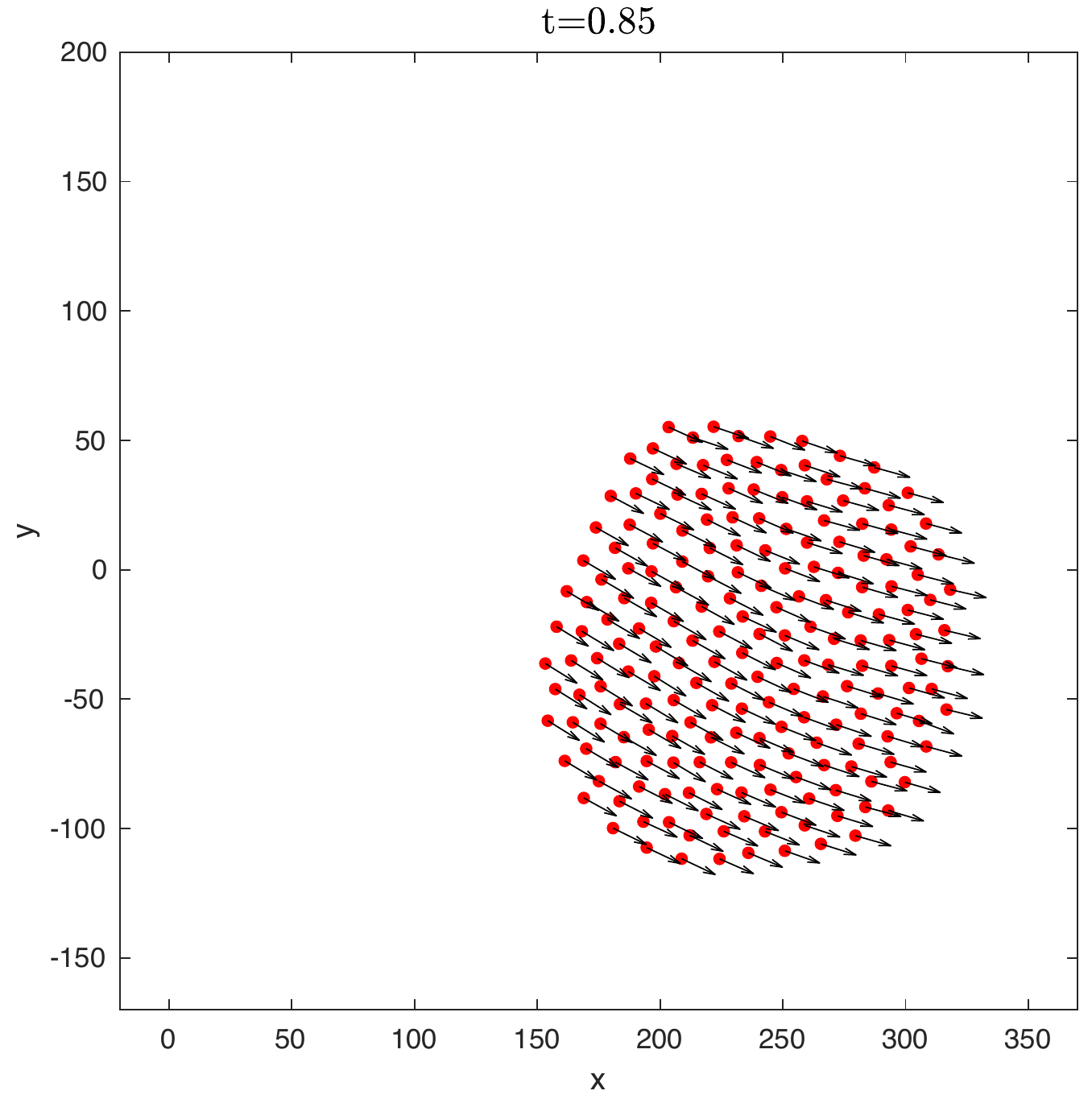}\hfill
    \textbf d.\includegraphics[width=0.47\textwidth]{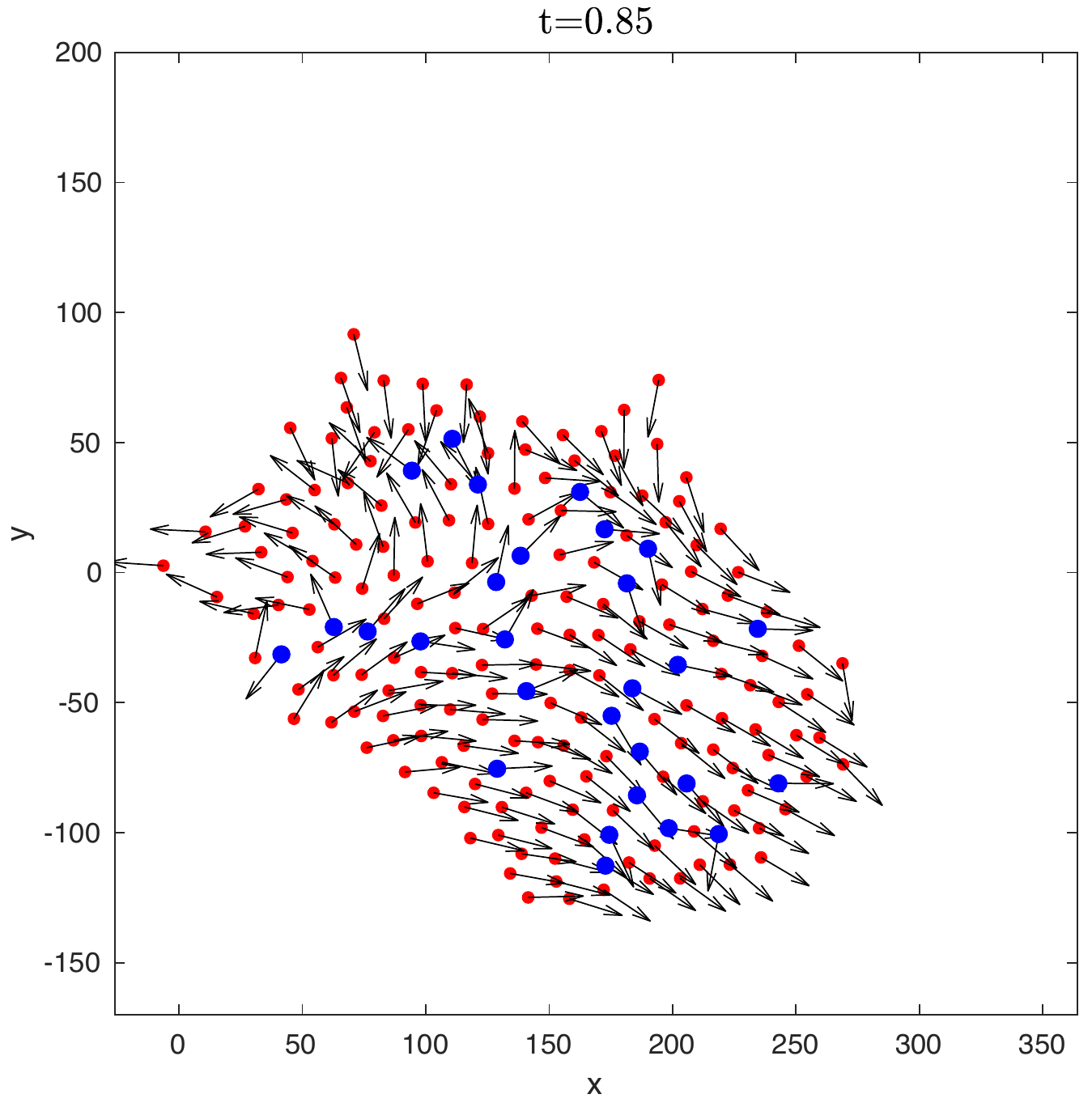}
    \caption{\revision{Test 2D: Screenshots of the numerical simulations without (\textbf a., \& \textbf c.) and with (\textbf b. \& \textbf d.) leaders. Red dots are followers, blue dots are leaders, and black arrows represent the normalized velocity. } }
    \label{test:new_test2D} 
\end{figure}

\newpage
The last test (Figure \ref{test:int_vs_ext_2D}) focuses on the difference between leaders flying at the boundary of the flock (Figure \ref{test:int_vs_ext_2D}\textbf {a}) and leaders flying in the interior of the flock (Figure \ref{test:int_vs_ext_2D}\textbf {b}). 
In order to show the turning attempt and its effects, it is convenient to look at the acceleration of the agents across the follower$\to$leader transition. 
In Figure \ref{test:int_vs_ext_2D}\textbf {c} we see that, once an external agent becomes leader, it is strongly repulsed outside the flock and experiences a strong acceleration. 
In Figure \ref{test:int_vs_ext_2D}\textbf {d}, instead, we see that an internal leader remains trapped in the flock, since all-around repulsion forces lead, together, to negligible effects.
\begin{figure}[h!]
    \centering
    \textbf a.\includegraphics[width=0.47\textwidth]{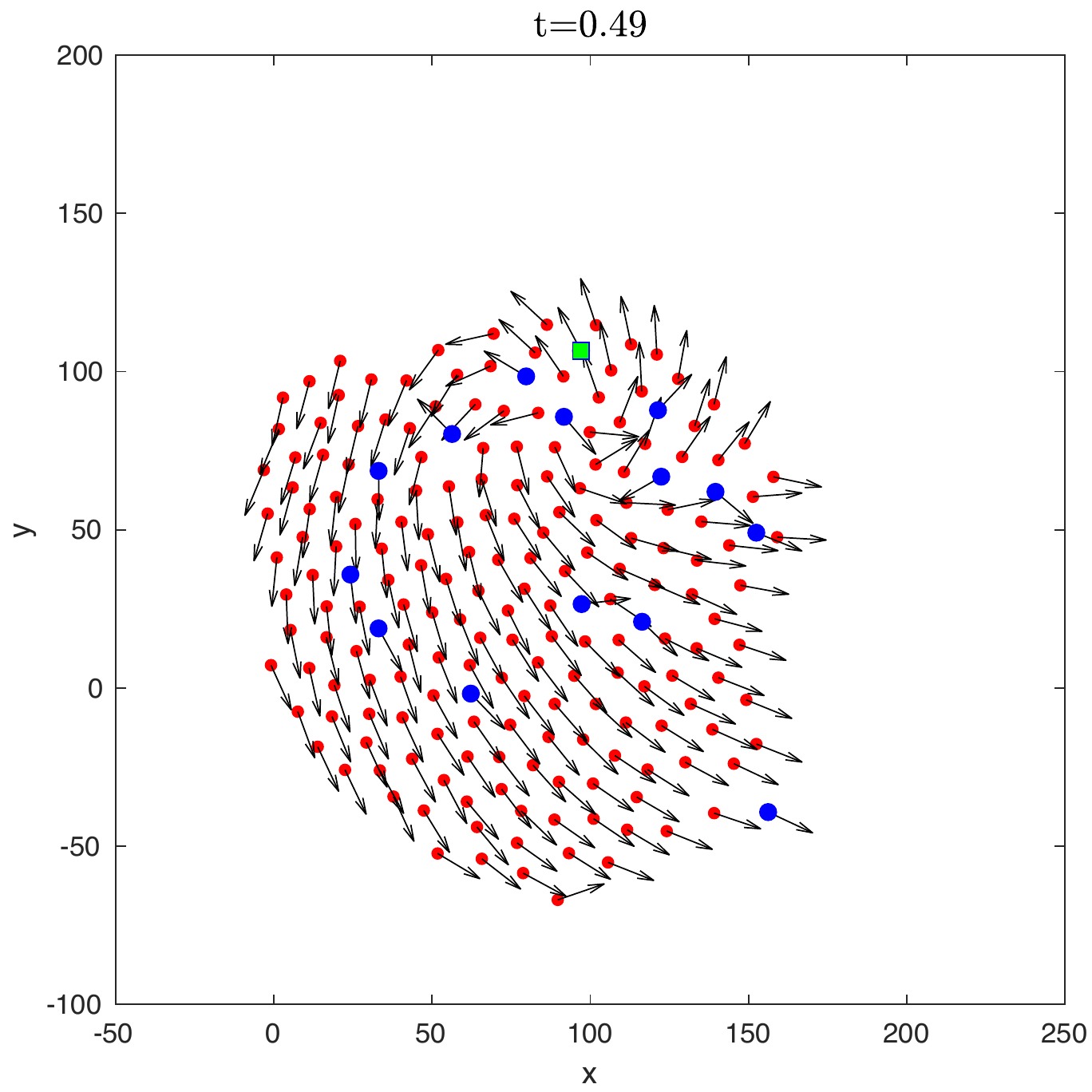}\hfill
    \textbf b.\includegraphics[width=0.47\textwidth]{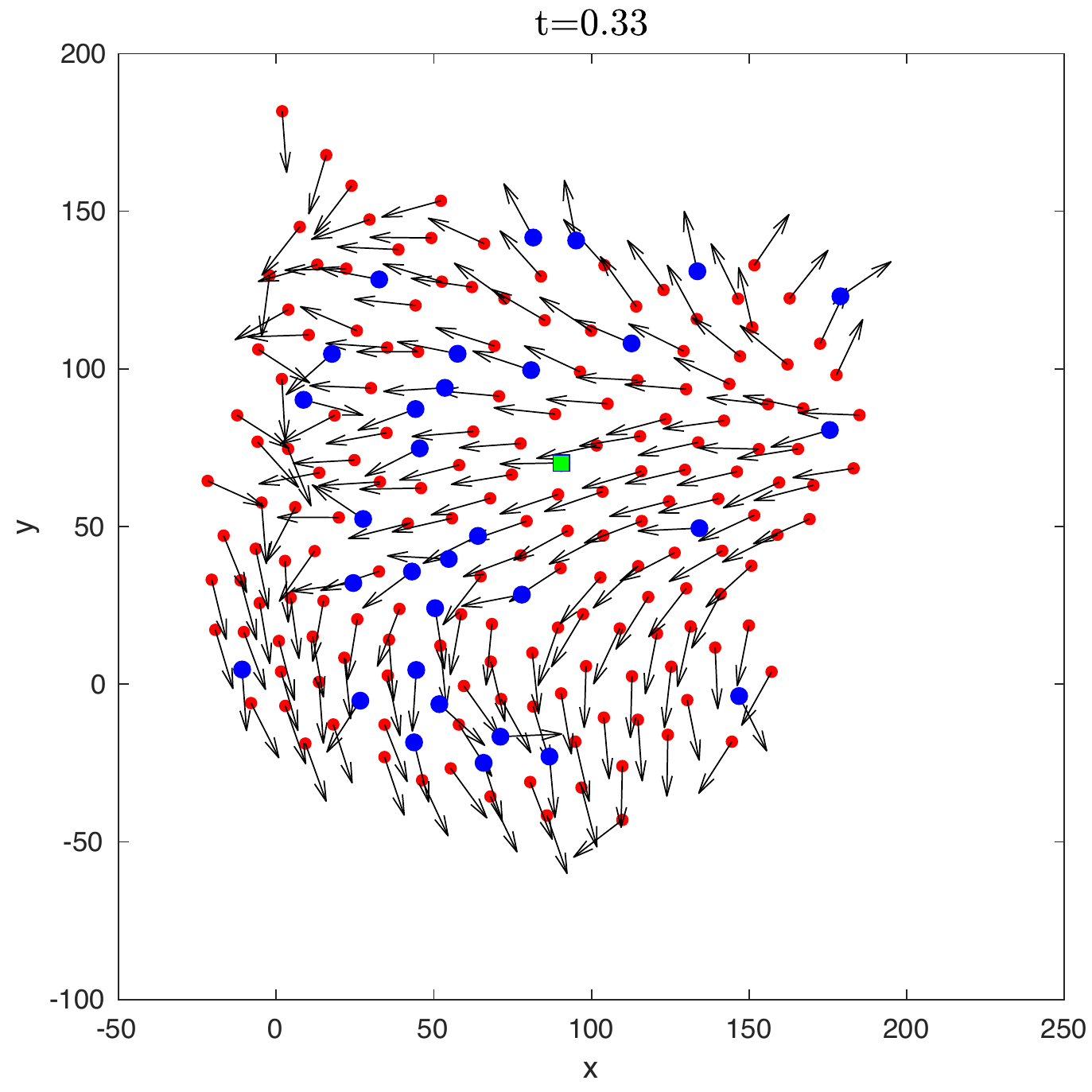} 
    \\
    \textbf c.\includegraphics[width=0.47\textwidth]{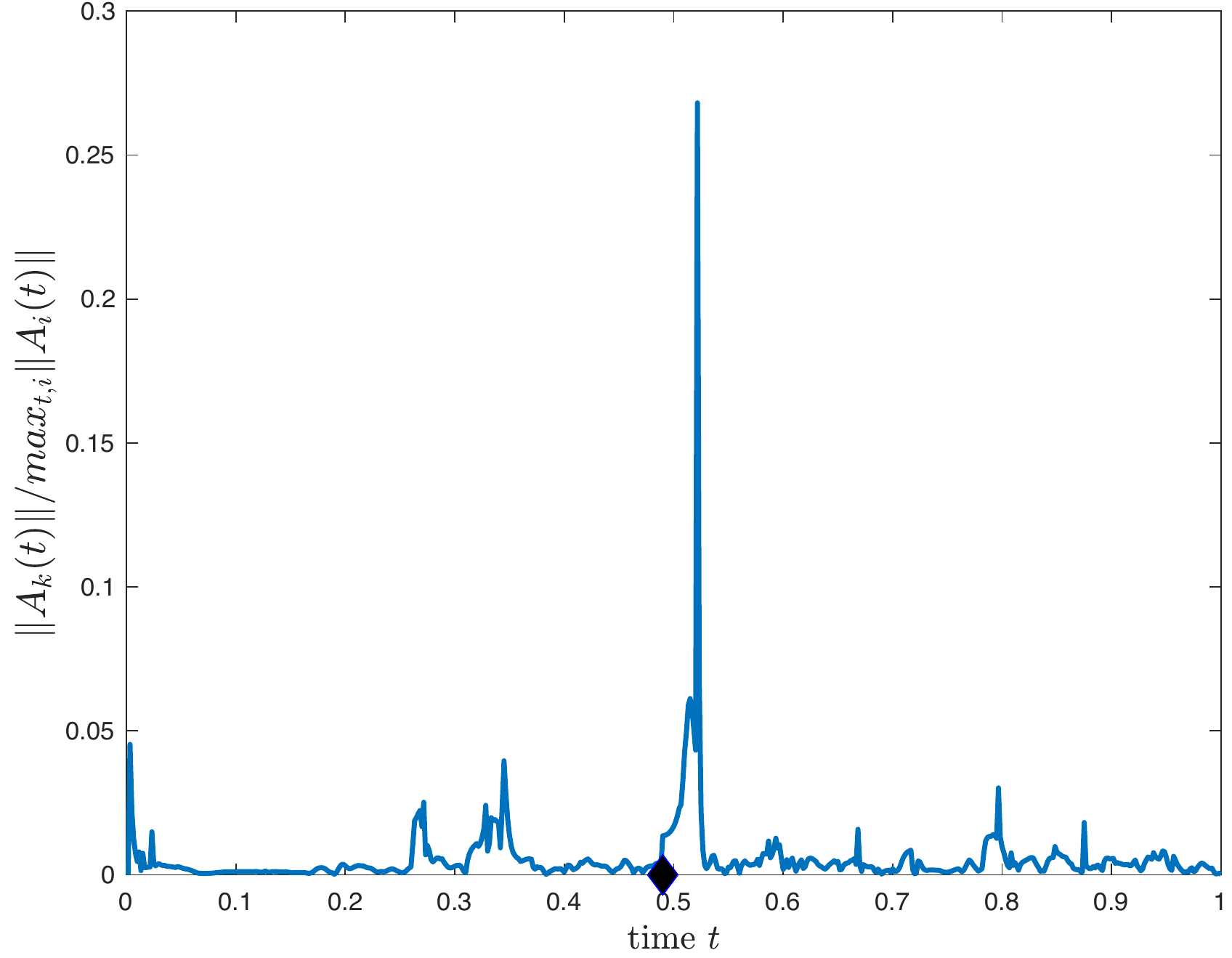}\hfill
    \textbf d.\includegraphics[width=0.47\textwidth]{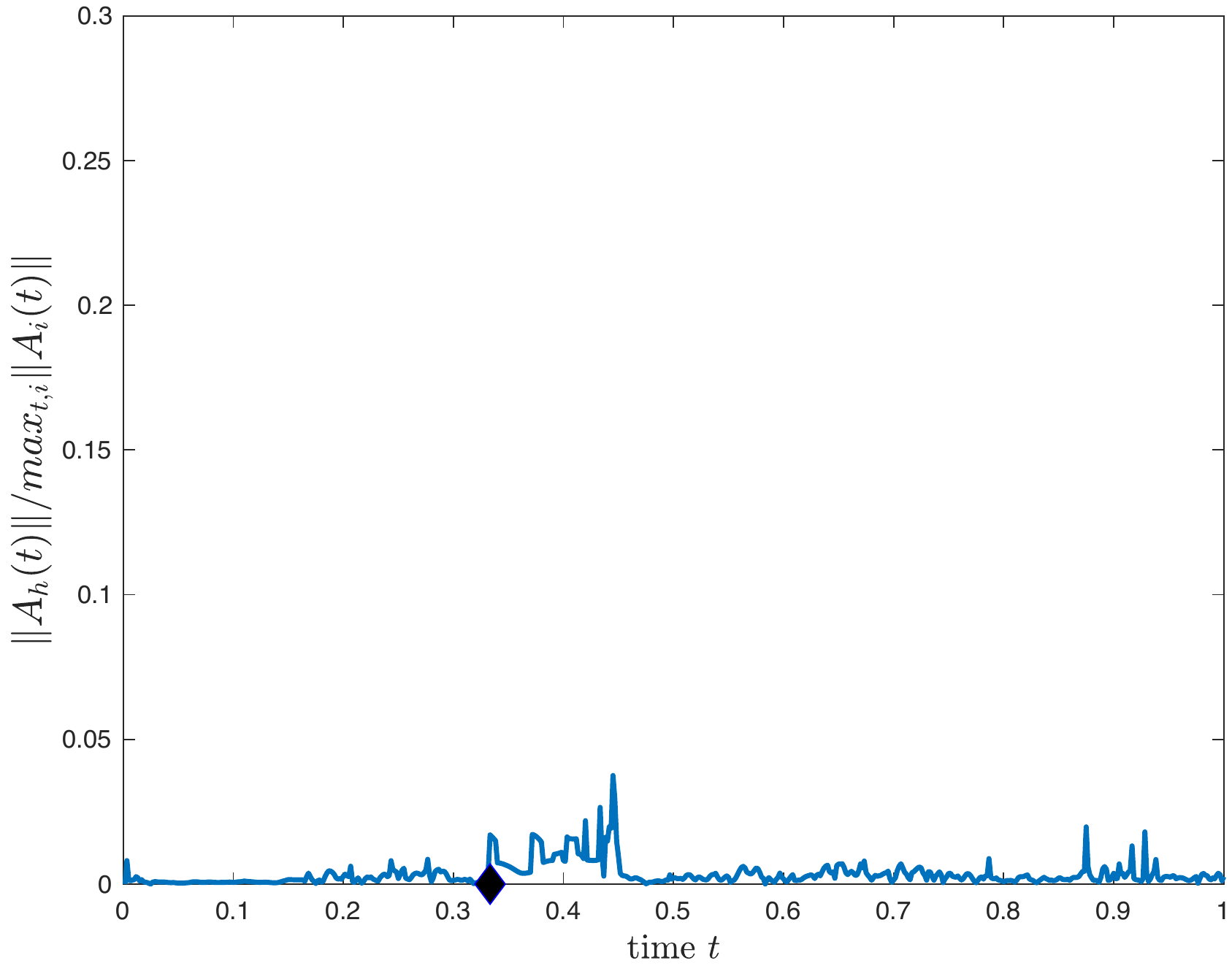}
    \caption{Test 2D: \textbf a. \& \textbf b. The flock at two instants of time. Red dots are followers, blue dots are leaders, green dots are selected leaders under observation. Arrows represent the normalized velocity. 
    In \textbf a., the selected leader \revision{$k$} is on the boundary of the flock, while in \textbf b. the selected leader \revision{$h$} is internal. 
    \textbf c. \& \textbf d. The modulus of the acceleration of the selected leaders in \textbf a. and \textbf b., respectively. The black diamond denotes the beginning of the leadership. The values of the modulus are normalized w.r.t.\ the maximum modulus observed during the entire simulation among all agents.} 
    \label{test:int_vs_ext_2D} 
\end{figure}
\clearpage

\subsection{A small 3D flock}
Here we present a 3D test with 400 agents. 

\revision{Figure \ref{test:screenshots3Dsmall}\textbf{a} and Figure \ref{test:screenshots3Dsmall}\textbf{c} show three screenshots of the moving flock. 
In the middle of the simulation (i.e.\ excluding the initial and the final fast turns) it is clearly visible a downward turn, which causes the flock to move downward for a while. 
This turn is caused by the combined effect of several leaders and it is better investigated in the Figure \ref{test:staffetta3D}. 
This figure shows the function $V_k^3(t)$ for six agents which -- by chance -- become leader one after the other and are all located at the bottom of the flock. 
As soon as they become leader, their $z$-component $V^3$ of the velocity immediately decreases, since they start pointing downward. 
Once they return to be followers, they slowly come back to the flock.
Figure \ref{test:screenshots3Dsmall}\textbf{d} shows the trajectory of the flock's barycenter. In particular, the yellow part, corresponding to the leadership period of the investigated leaders, highlights their crucial role in triggering the overall downward turn.
}
%
\begin{figure}[h!]
    \centering
    \textbf a.\includegraphics[width=0.477\textwidth]{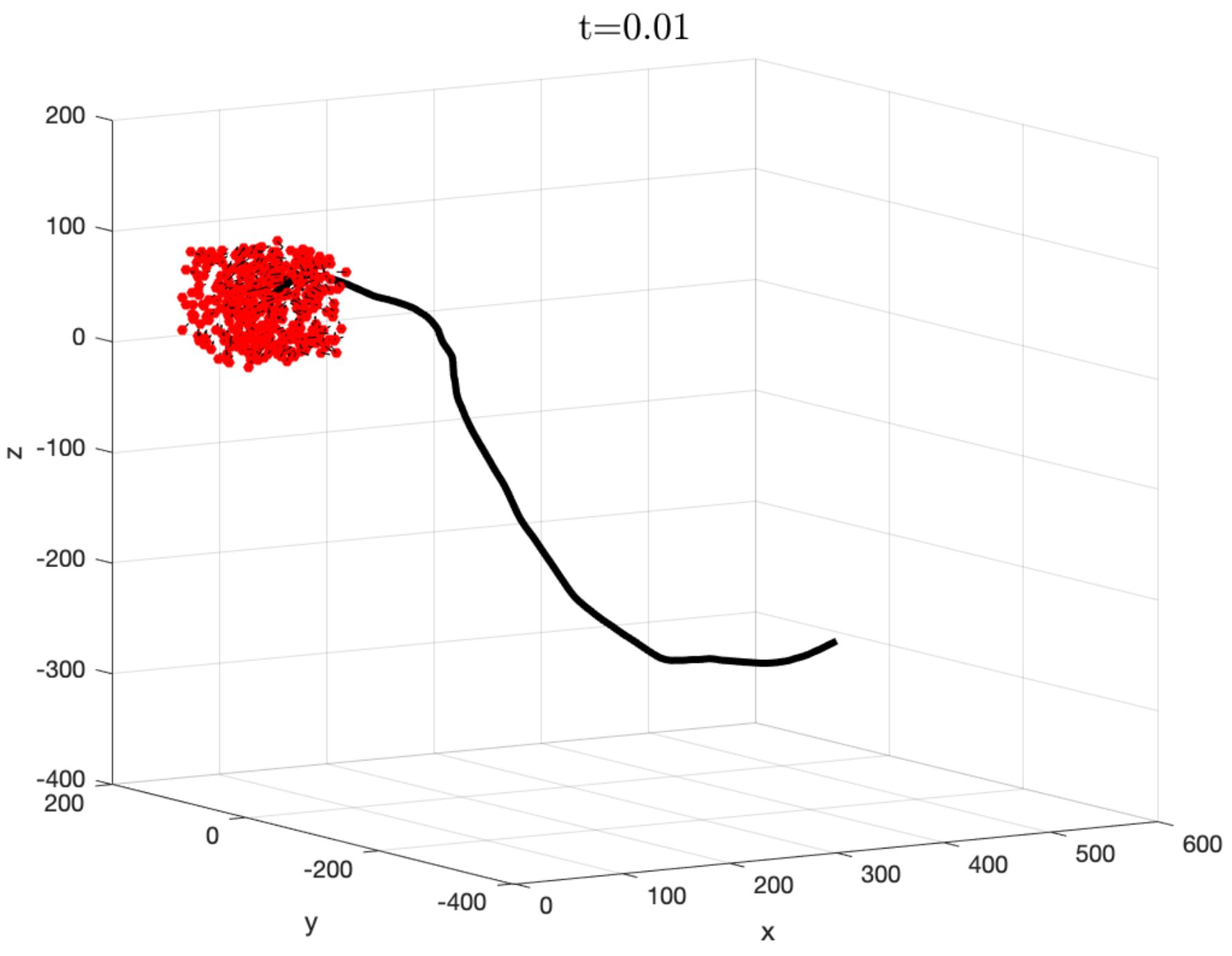}\hfill
    \textbf b.\includegraphics[width=0.477\textwidth]{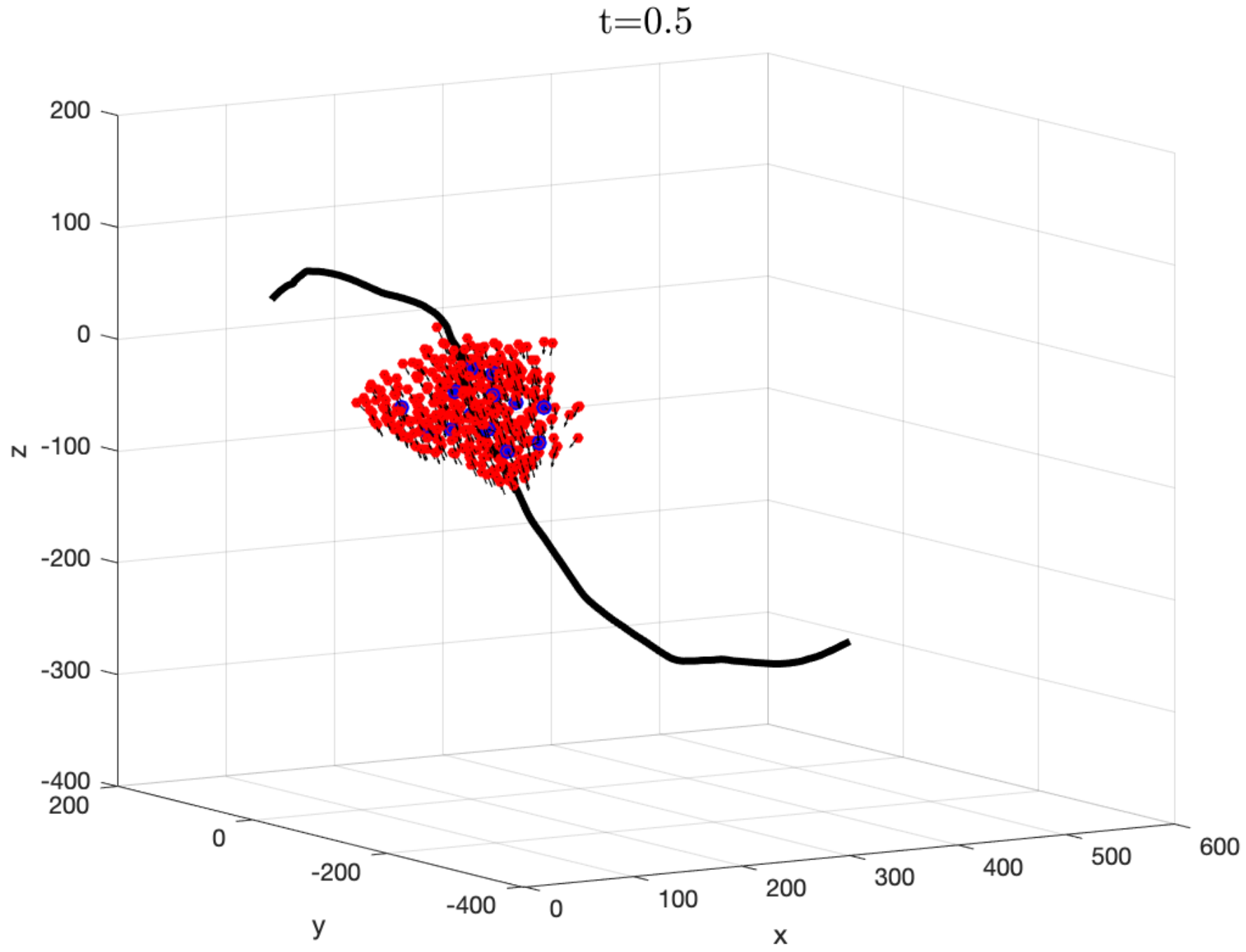} 
    \\
    \textbf c.\includegraphics[width=0.477\textwidth]{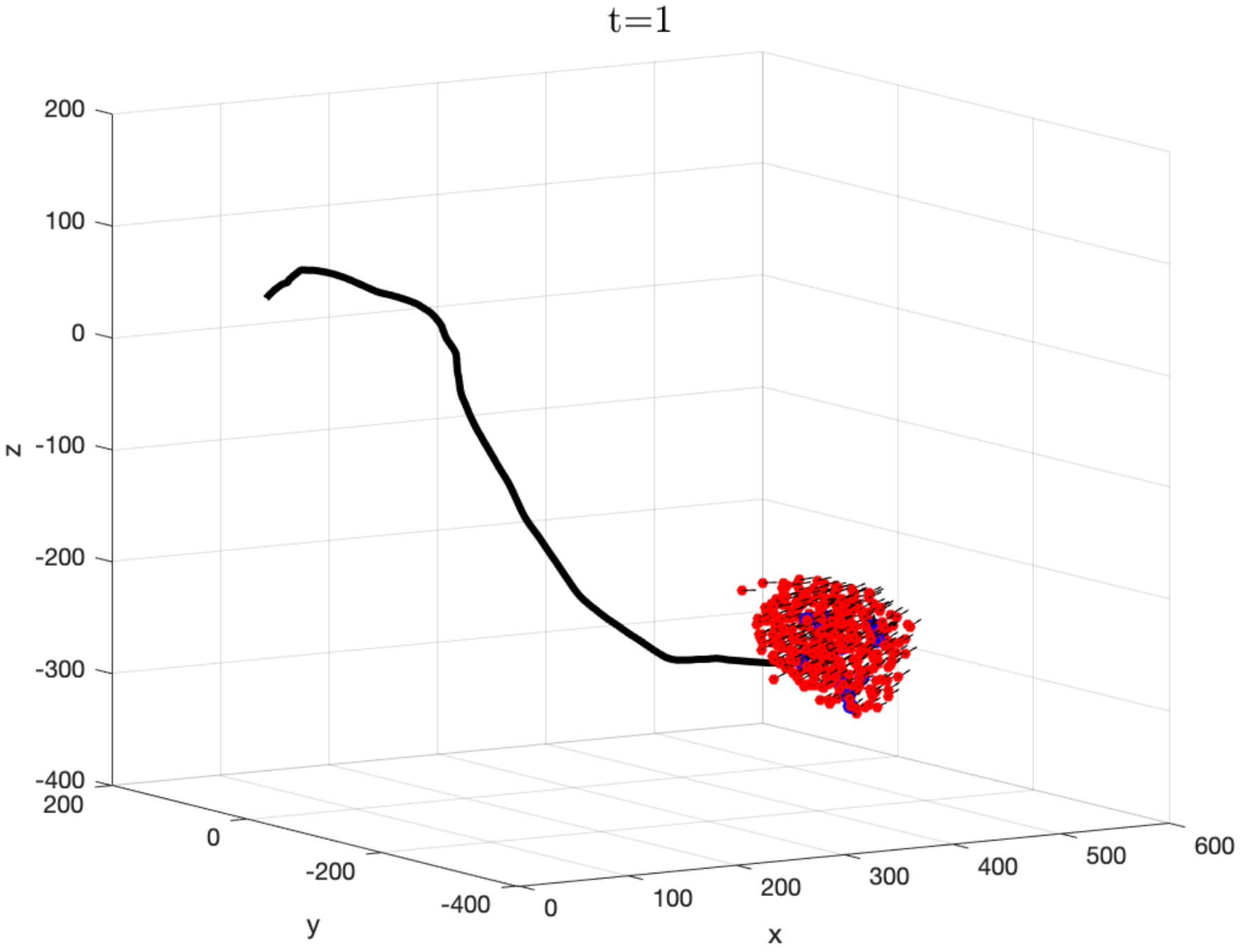}\hfill
    \textbf d.\includegraphics[width=0.477\textwidth]{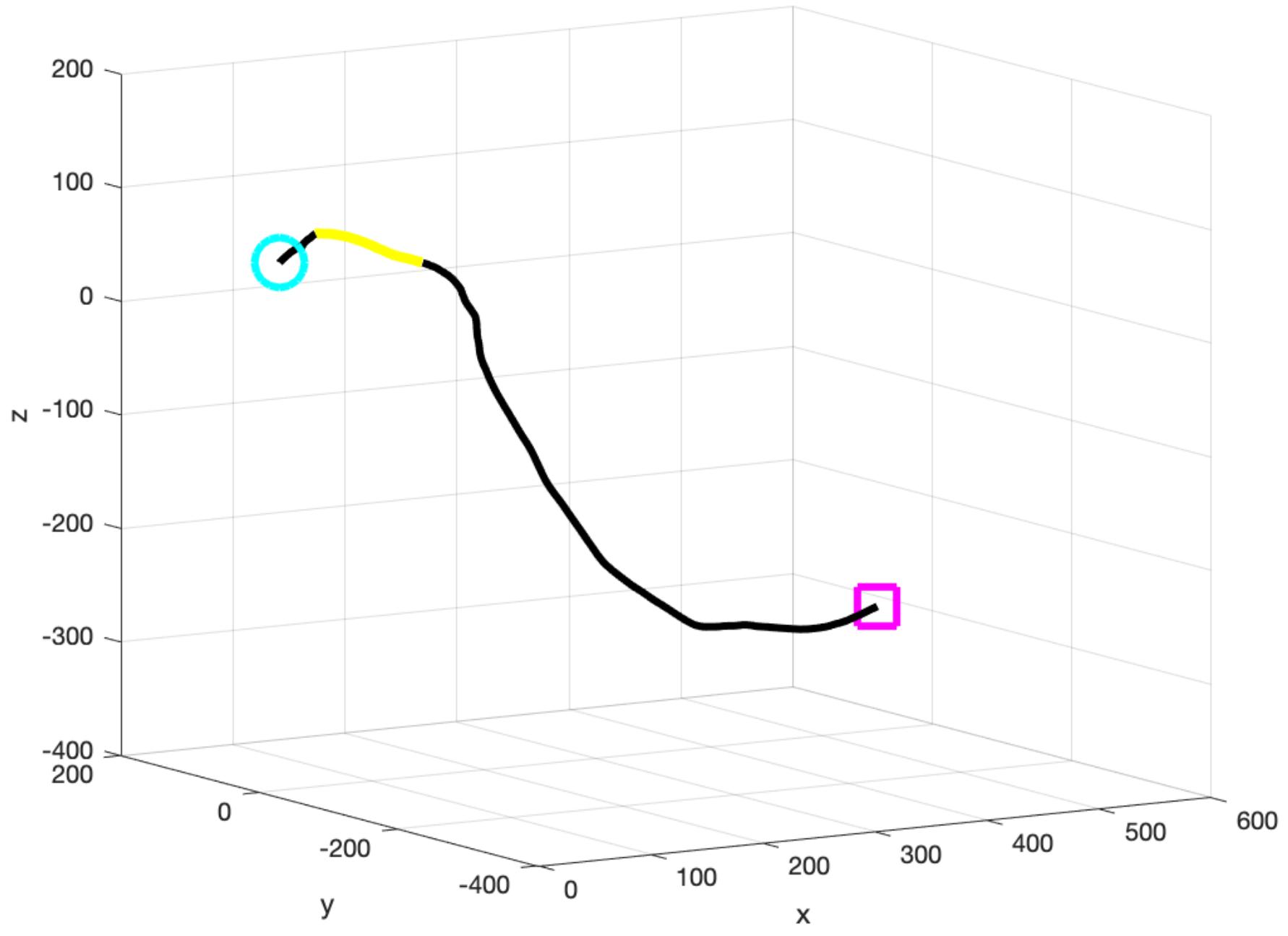}
    \caption{Test \revision{3D-small}: \textbf a., \textbf b. \& \textbf c. Three screenshots of the moving flock. Black line is the trajectory of the barycenter. 
    \revision{\textbf d. Trajectory of the barycenter: cyan circle and magenta square denote, respectively, the initial and final positions of flock barycenter. The triggering phase of the downward turn is marked in yellow. }
    Animated simulation available at \url{www.emilianocristiani.it/attach/starlings-3Dsmall.mov}
    } 
    \label{test:screenshots3Dsmall} 
\end{figure}
%
%

\clearpage

\begin{figure}[h!]
    \centering
    \includegraphics[width=0.72\textwidth]{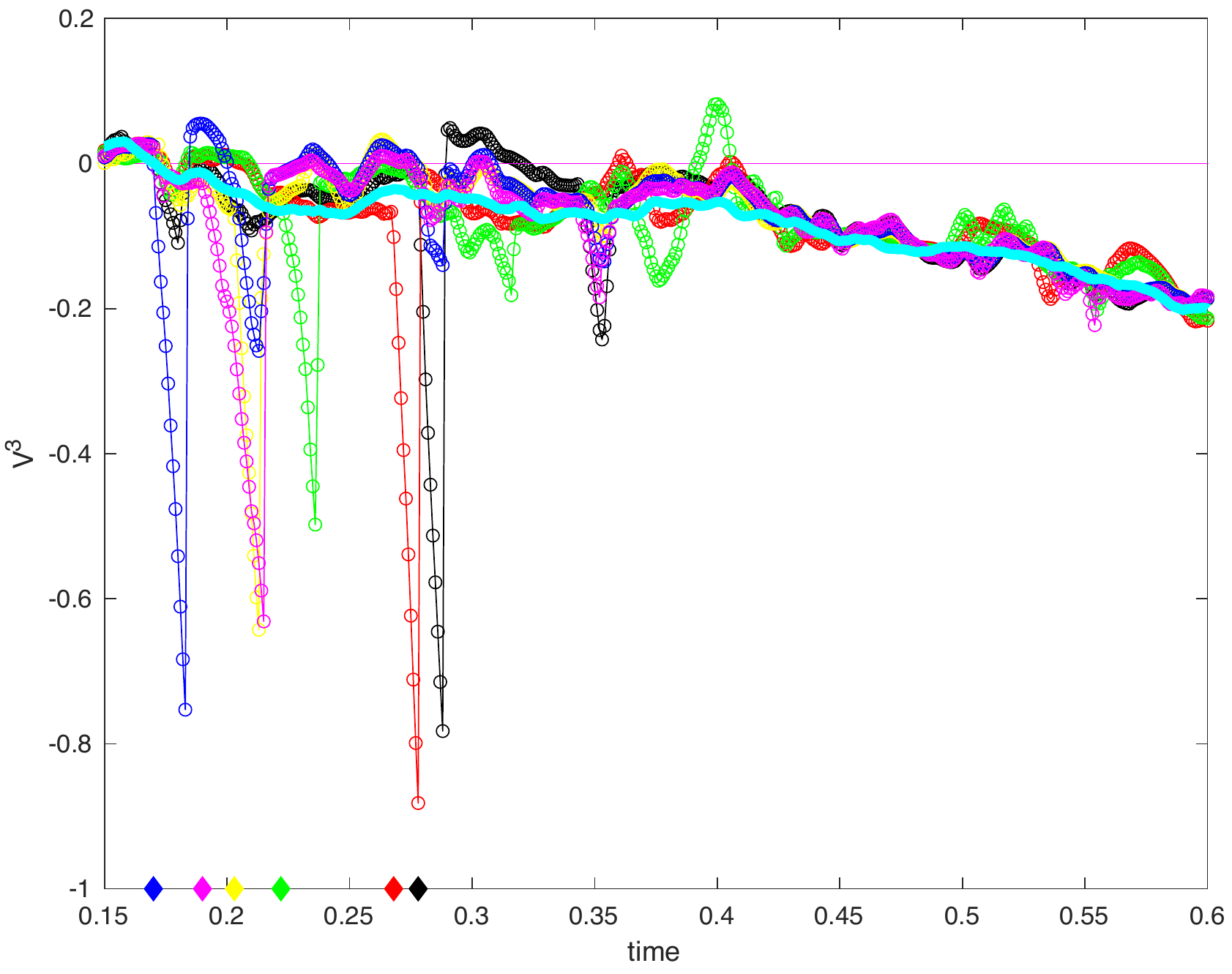}
    \caption{Test \revision{3D-small}: $z$-component $V^3$ of the velocity of six agents all located at the bottom of the flock which becomes leader one after the other. Diamonds denote the time agents become leaders. \revision{The light blue curve represents the mean  $z$-component of the velocities of the followers.}}
    \label{test:staffetta3D} 
\end{figure}
Figure \ref{test:turningtrigger} shows the flock during the leadership period of one of the six agents mentioned before. It is clearly visible its contribution to the downward turn and its effect  among its neighbors.
\begin{figure}[h!]
    \centering
    \includegraphics[width=0.73\textwidth]{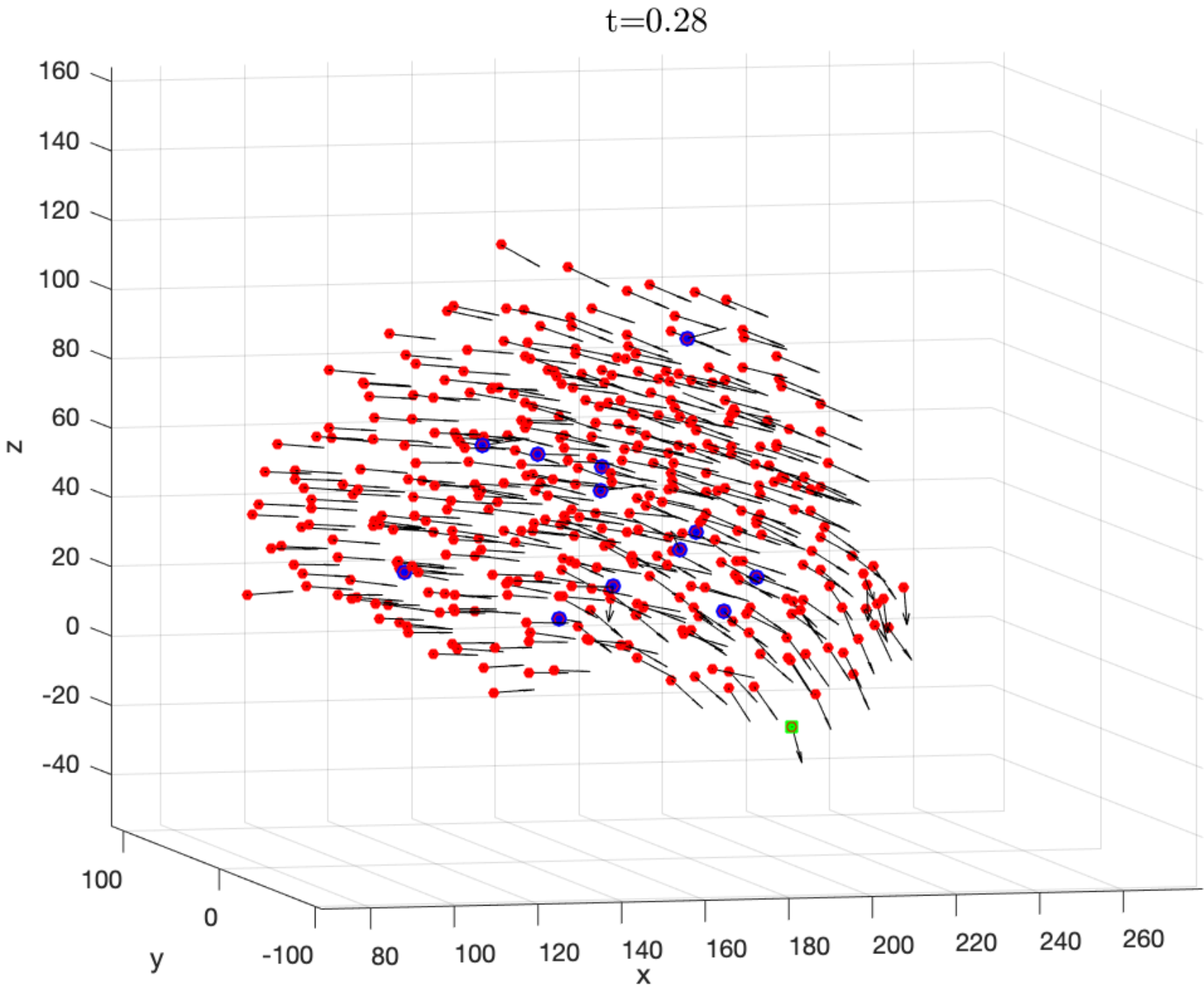}
    \caption{Test \revision{3D-small}: turn triggered by the six leaders shown in Figure \ref{test:staffetta3D}. Green square denotes the agent, among the six, which is leader at this time and is triggering (better, reinforcing) the downward turn.}
    \label{test:turningtrigger} 
\end{figure}

\clearpage

\subsection{A large 3D flock}
Here we present a 3D test with 2000 agents.
Figure \ref{test:screenshots3Dlarge} shows four screenshots of the moving flock, \revision{and Figure \ref{test:3Dlarge-bar-elong} shows the trajectory of the barycenter of the flock and its elongations in the directions $x$, $y$, and $z$. }
The flock stretches and compresses in a way directly comparable to the flocks in Figure \ref{fig:realstarlings}. 
In this case the flock is so large that more than one turn at the same time can be triggered, in different parts of the flock. 
\revision{At a certain time the flock seems to split, but then it reunites.}

\begin{figure}[h!]
    \centering
    \textbf a.\includegraphics[width=0.47\textwidth]{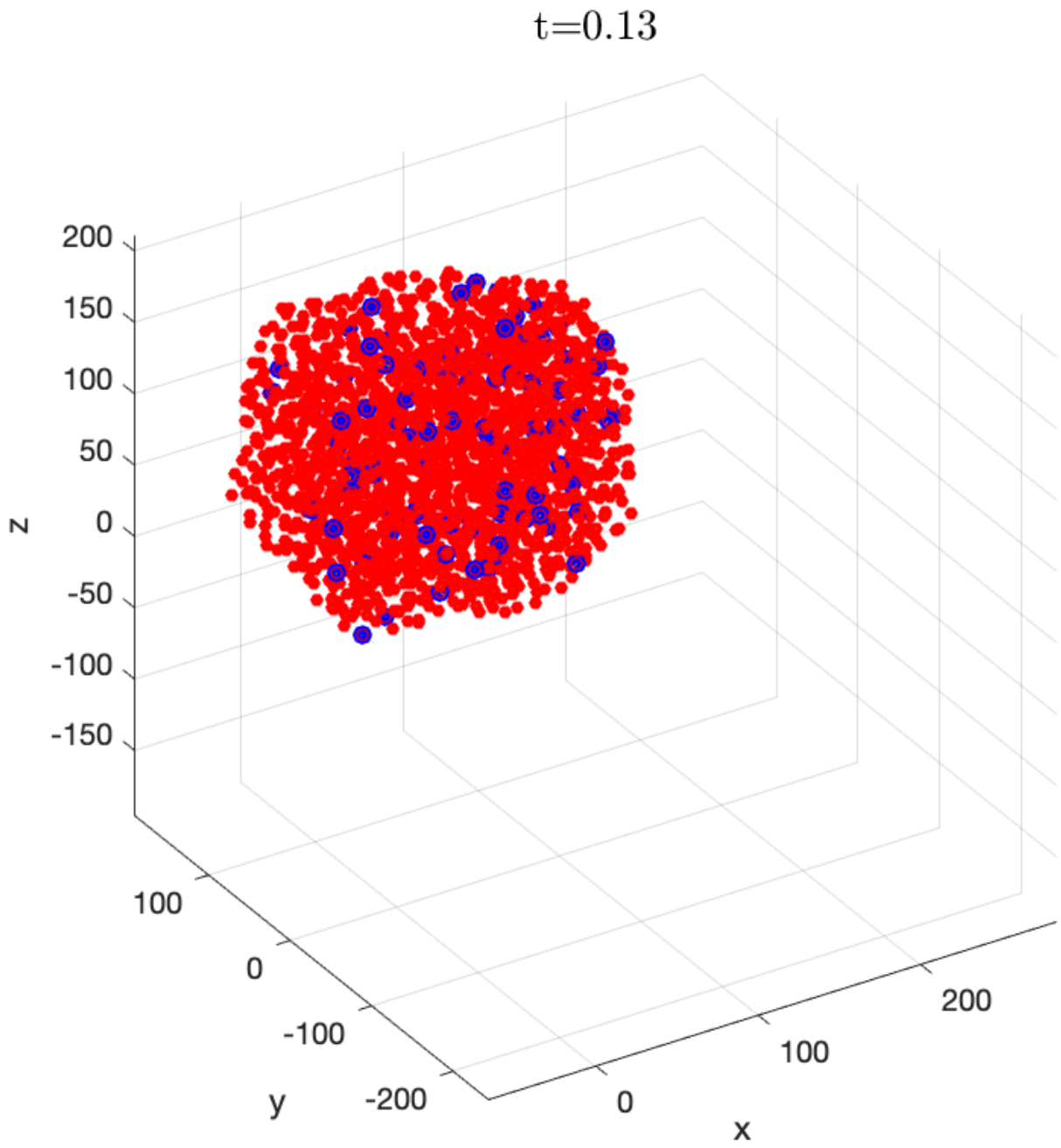}\hfill
    \textbf b.\includegraphics[width=0.47\textwidth]{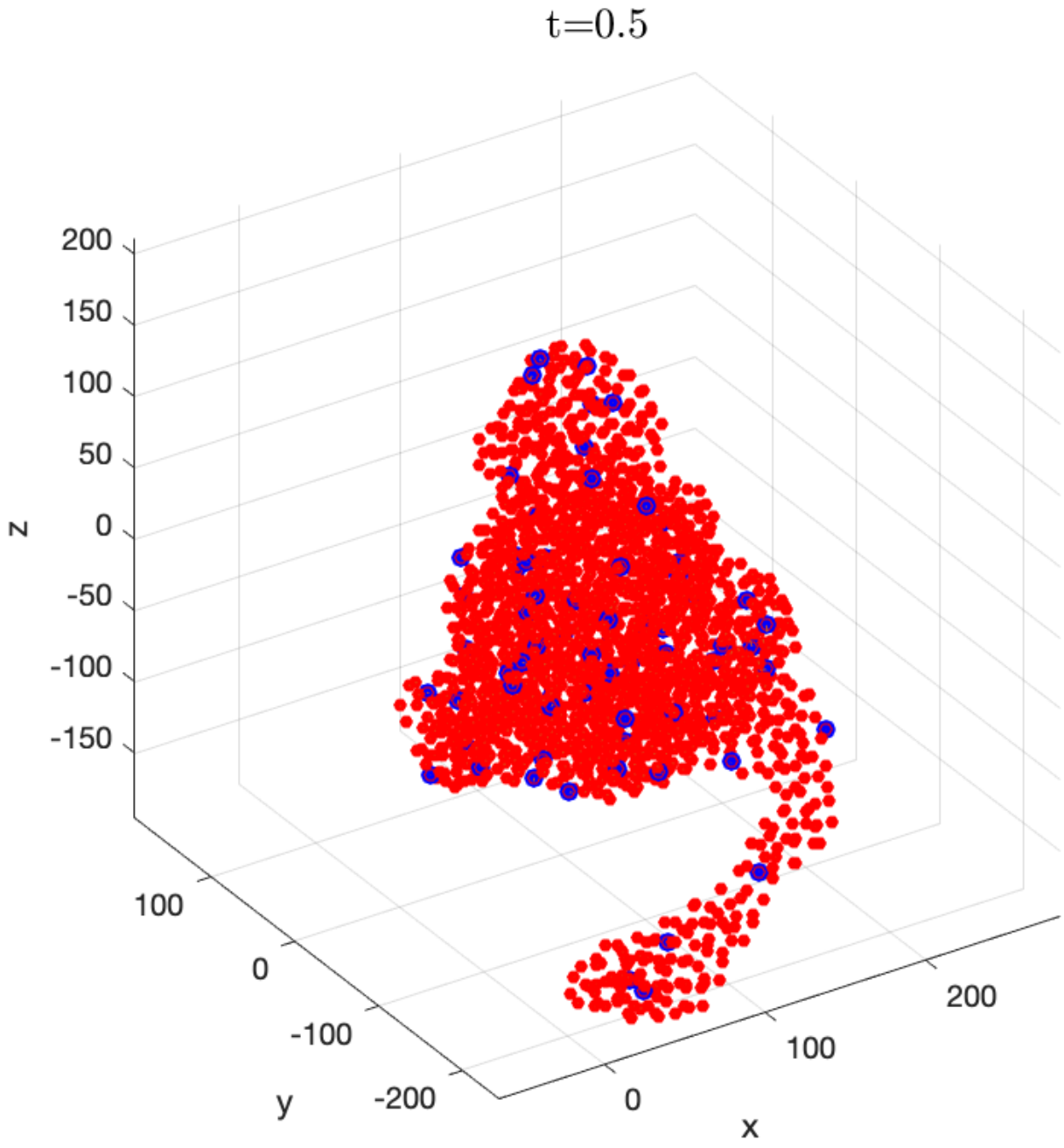} 
    \\
    \textbf c.\includegraphics[width=0.47\textwidth]{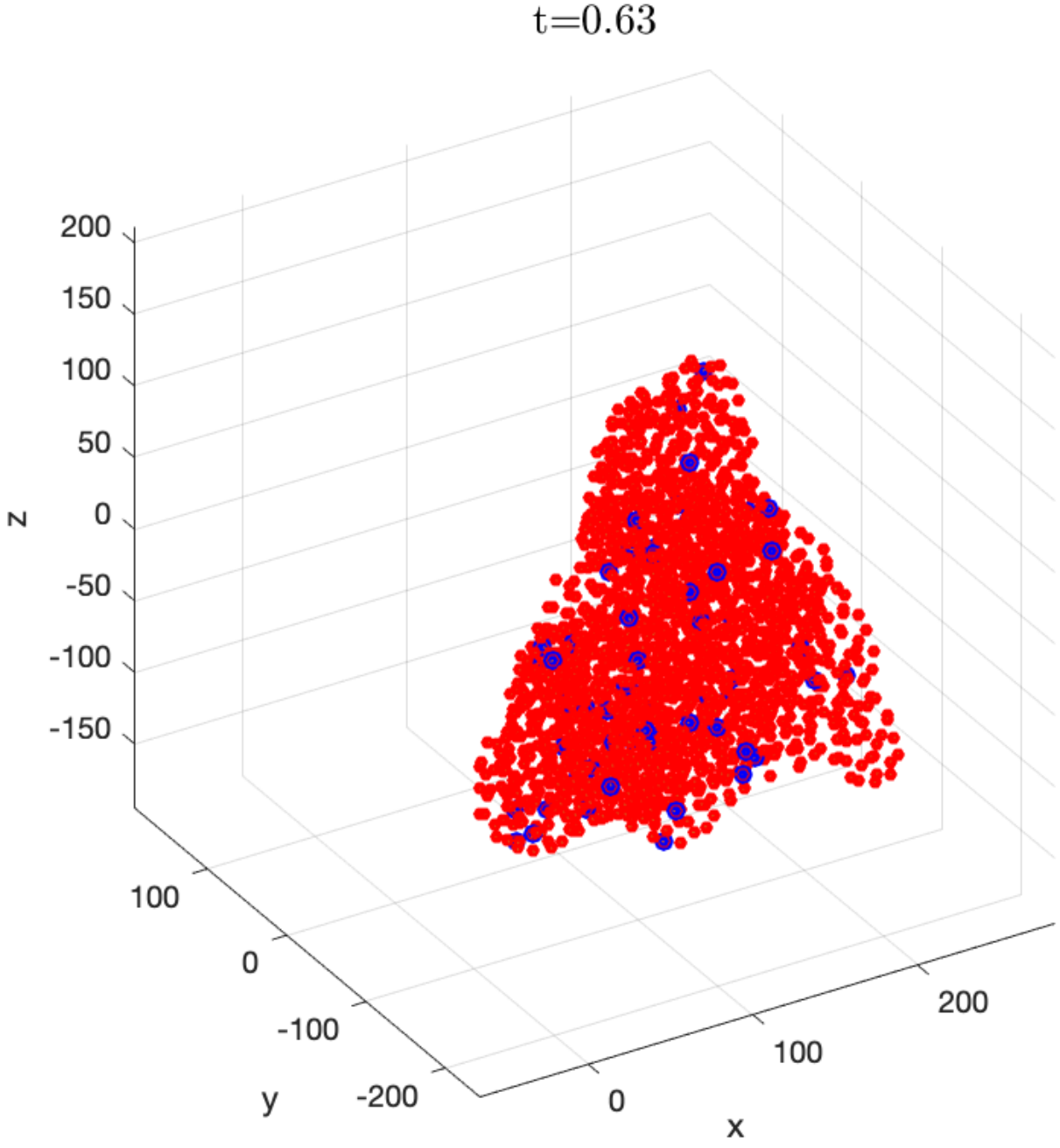}\hfill
    \textbf d.\includegraphics[width=0.47\textwidth]{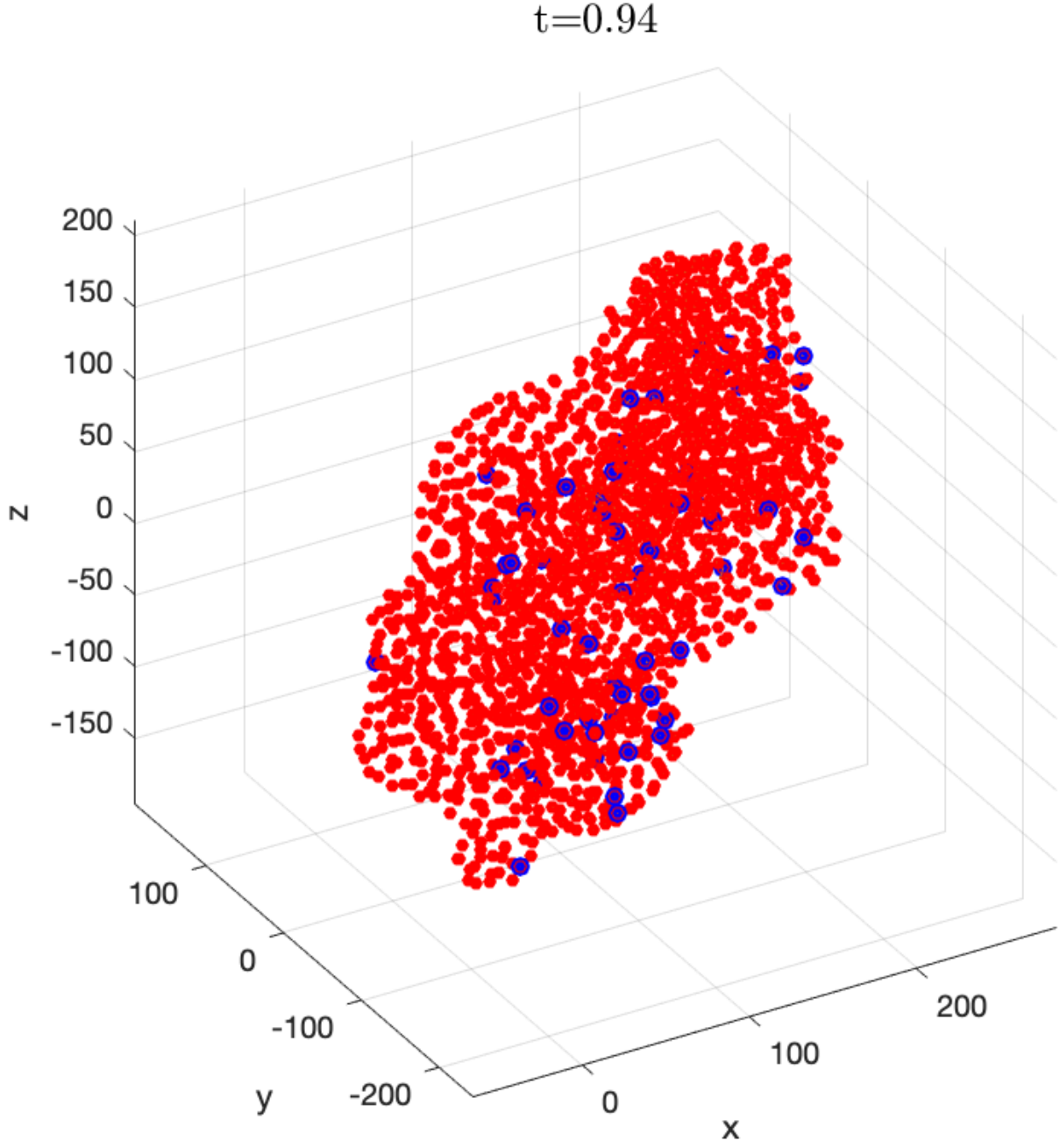}
    \caption{Test \revision{3D-large}: Four screenshots of the moving flock. The flock seems to split at a certain time (\textbf b.), but then it reunites. 
    Animated simulation available at \url{www.emilianocristiani.it/attach/starlings-3Dlarge.mov}
    }
    \label{test:screenshots3Dlarge} 
\end{figure}

\clearpage
\begin{figure}[h!]
    \centering
    \textbf a.\includegraphics[width=0.43\textwidth]{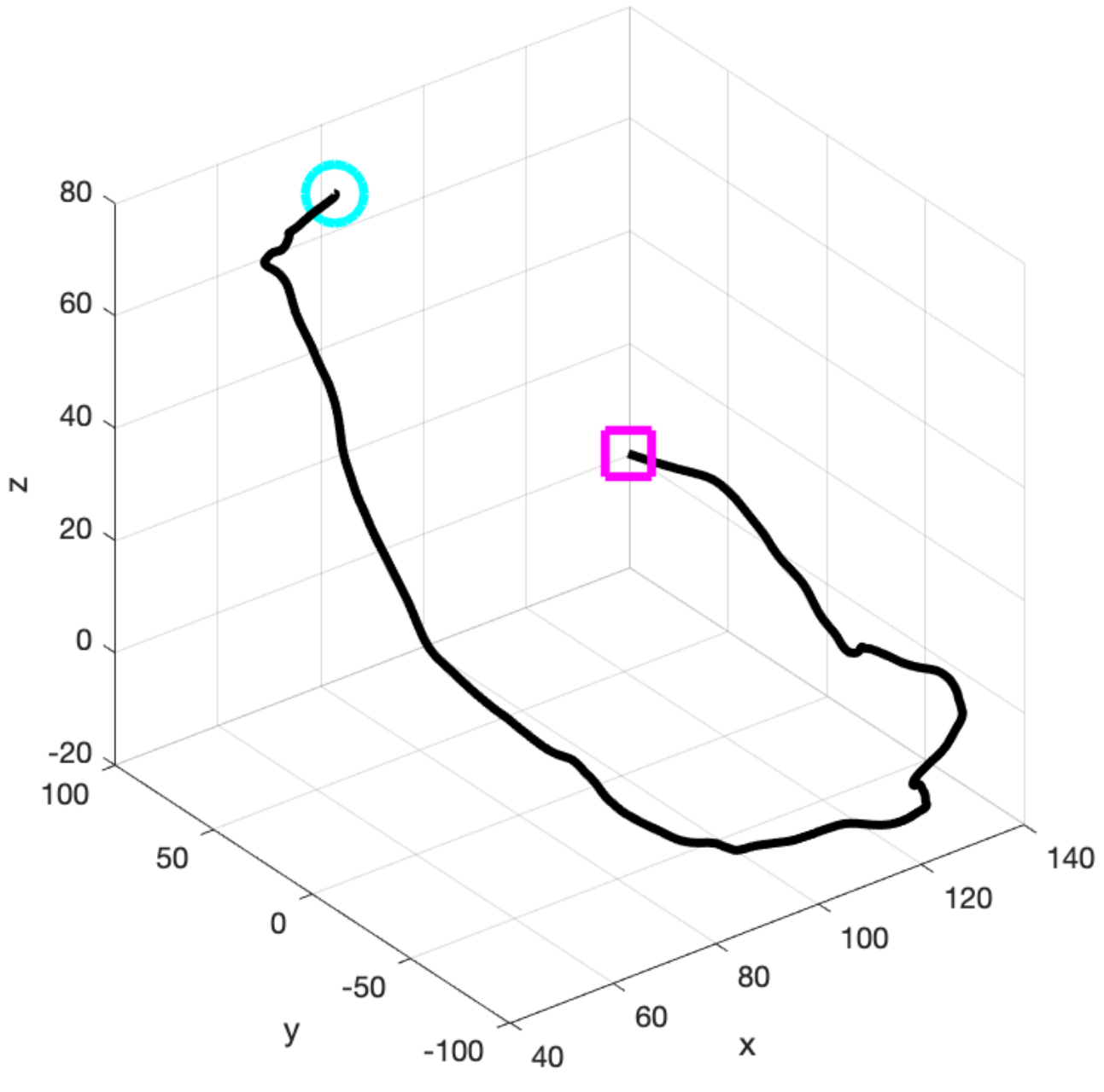}\hfill
    \textbf b.\includegraphics[width=0.49\textwidth]{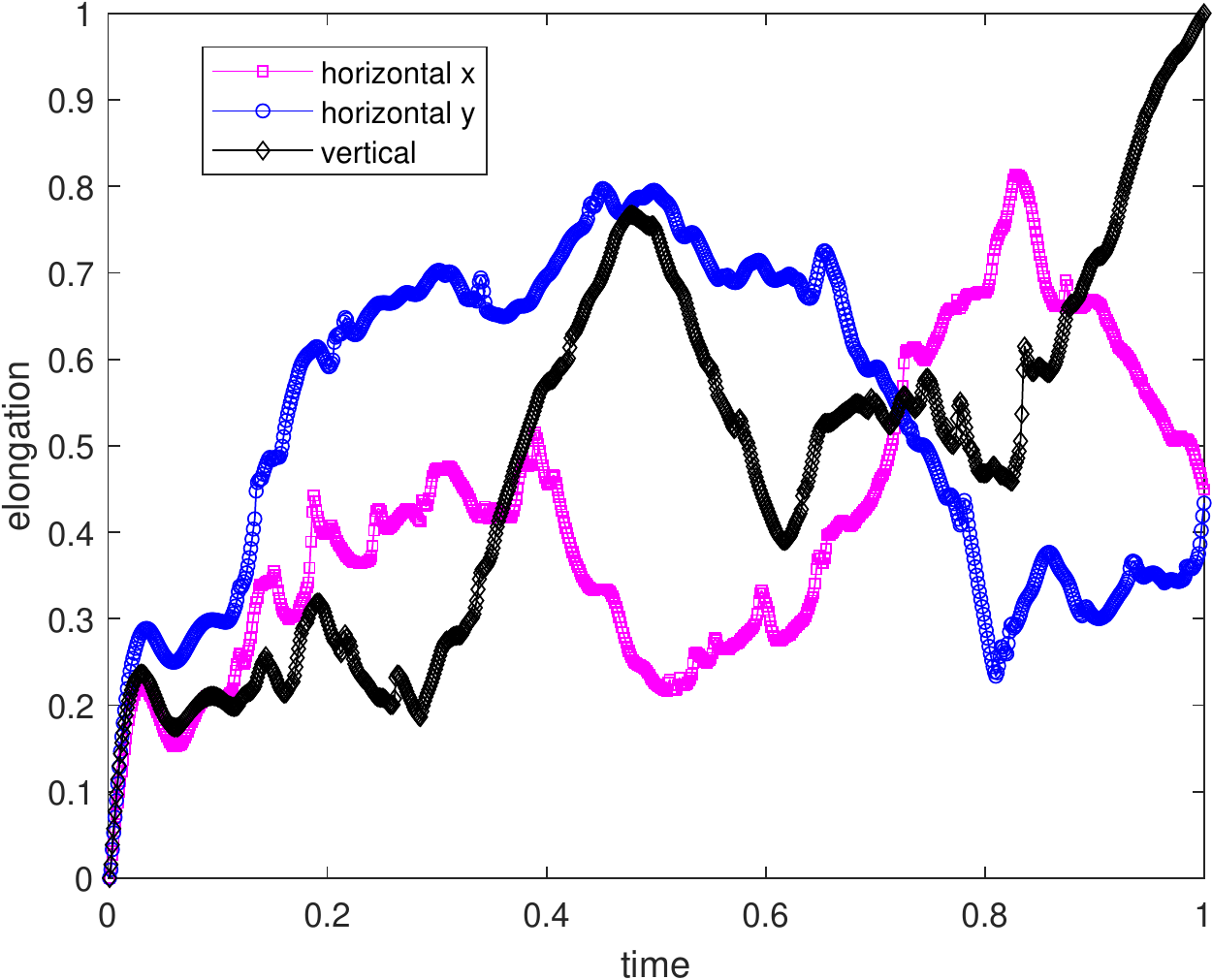} 
    \caption{Test \revision{3D-large}:  \revision{\textbf a. Trajectory of the barycenter: cyan circle and magenta square denote, respectively, the initial and final positions of flock barycenter.} \textbf b. Elongations of the flock in the three directions of the space. Elongations are measured in percentage with respect to the minimal and maximal elongations reached during simulation.}
    \label{test:3Dlarge-bar-elong} 
\end{figure}

\revision{Figure \ref{test:3Dlarge_revision} shows instead a simulation which results in a flock splitting. 
With respect to the parameters setting of the previous 3D test, we have increased the value of the probability transition to $5\times 10^{-4}$, and decreased the refractory time, setting $\frak r=200$, hence allowing a same agent to become leader again after a shorter time.
As a result, a higher number of contemporary turns in different directions appear and cause the division of the flock into two autonomous groups (see Figure \ref{test:3Dlarge_revision}\textbf c and Figure \ref{test:3Dlarge_revision}\textbf d).}

\begin{figure}[h!]
    \centering
    \textbf a.\includegraphics[width=0.47\textwidth]{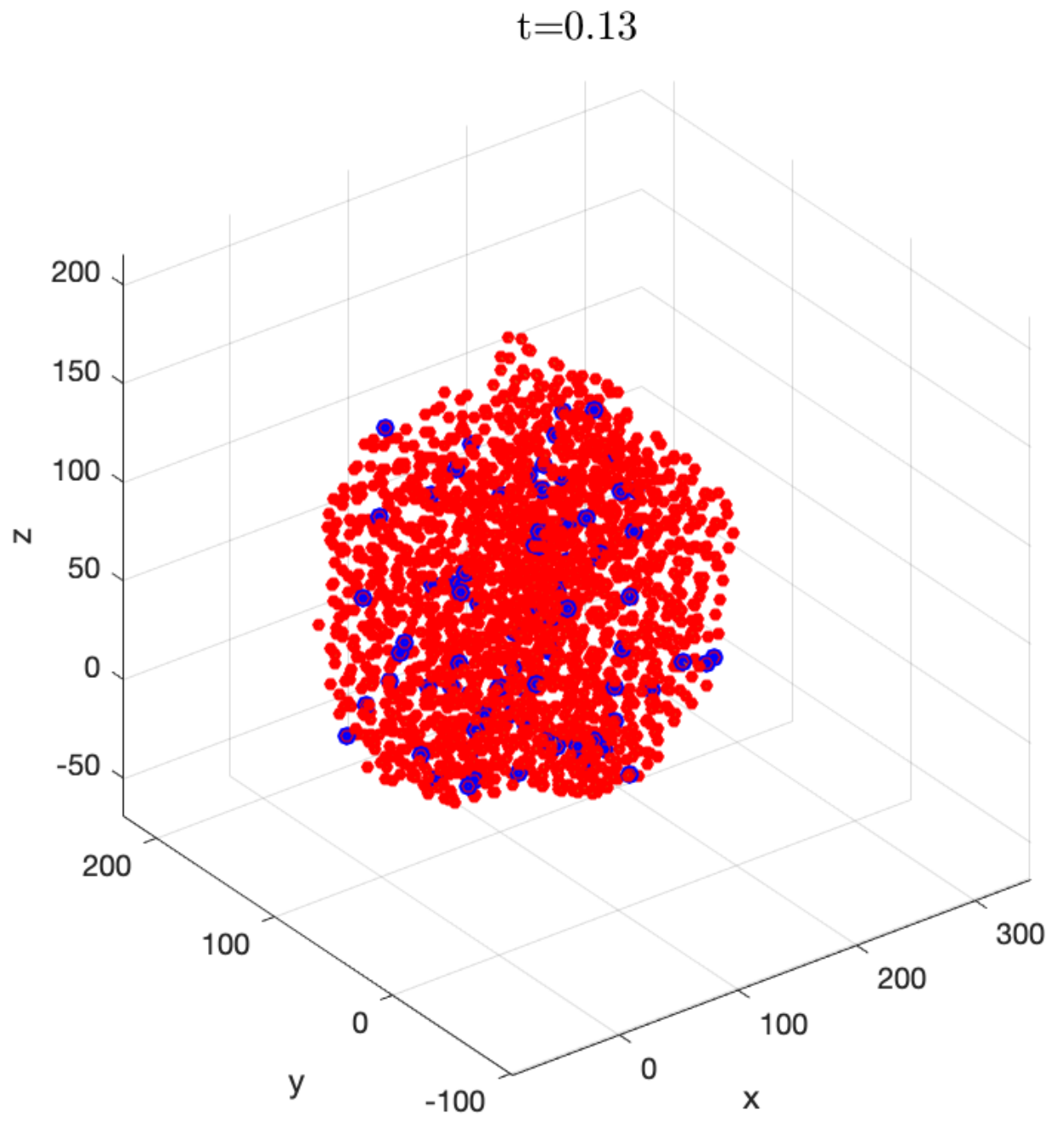}\hfill
    \textbf b.\includegraphics[width=0.47\textwidth]{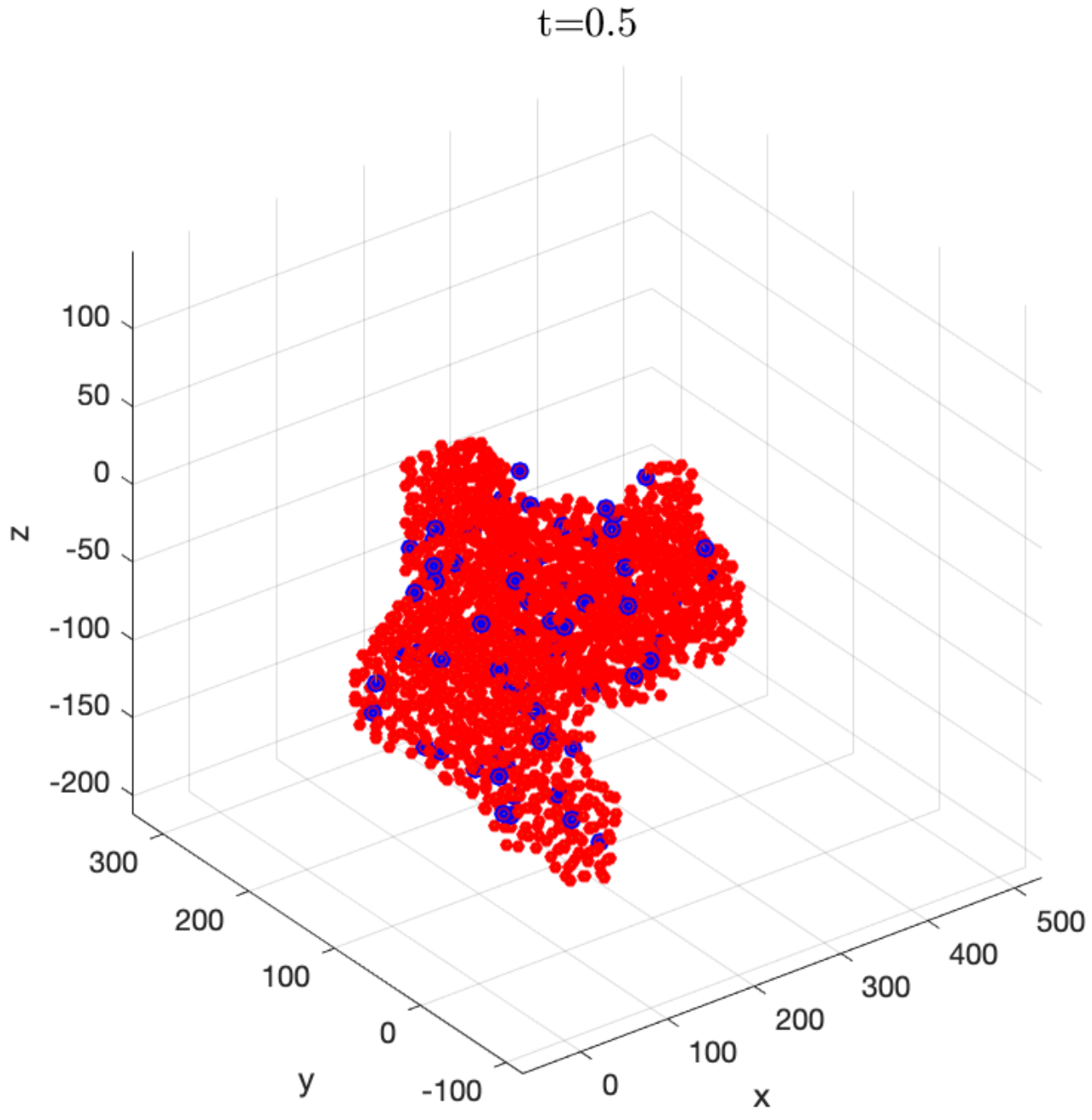} 
    \\
    \textbf c.\includegraphics[width=0.47\textwidth]{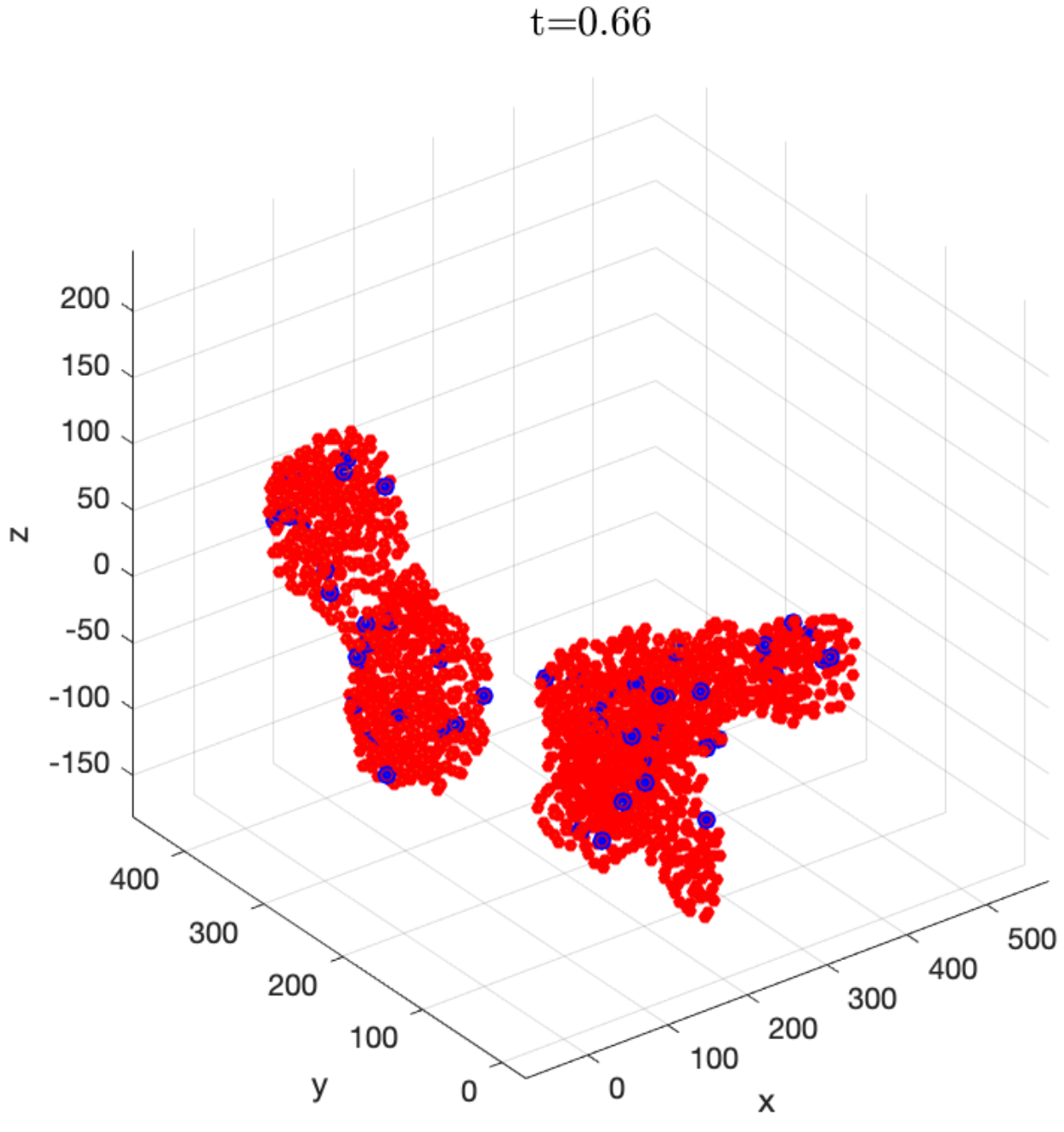}\hfill
    \textbf d.\includegraphics[width=0.47\textwidth]{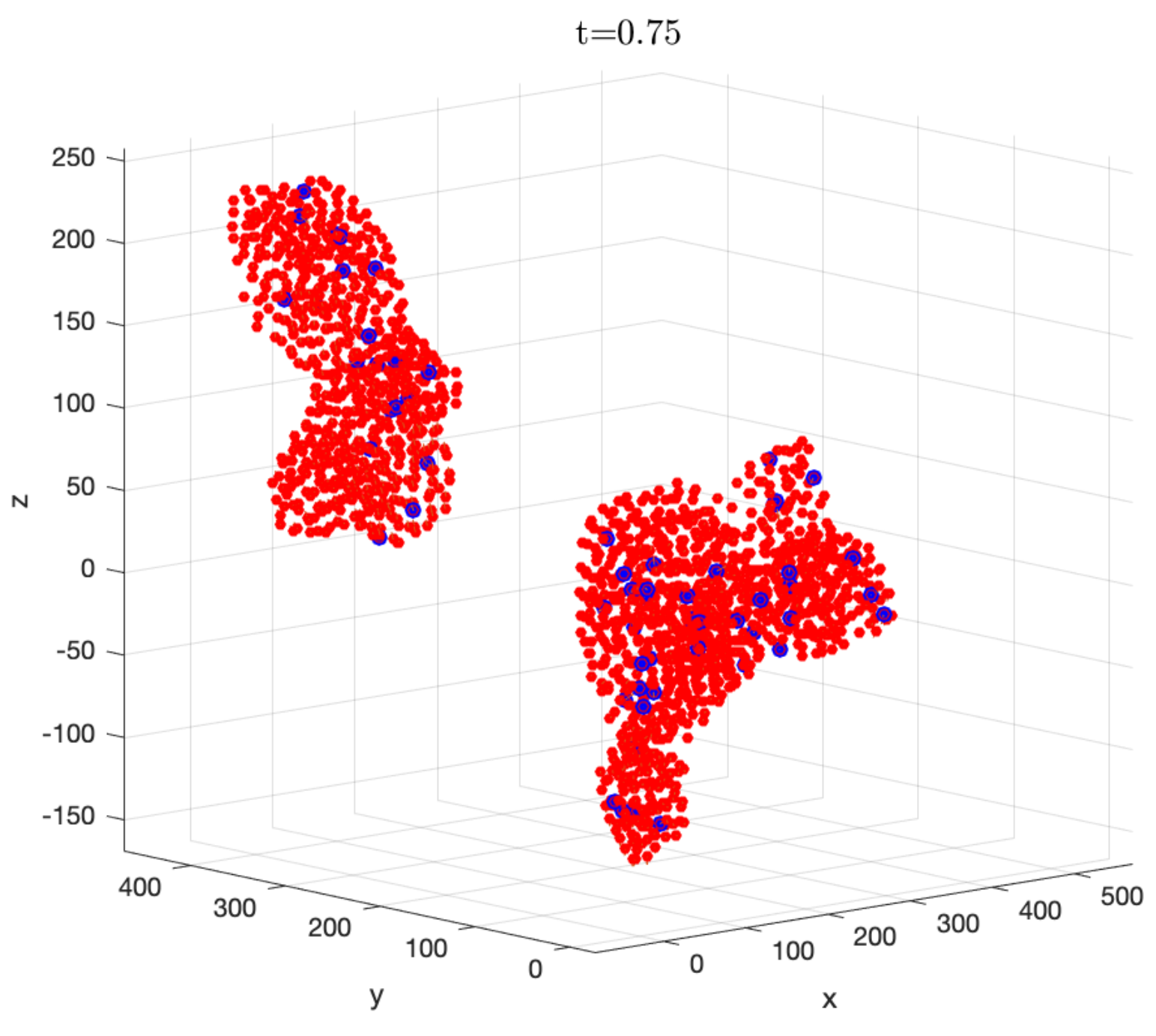}
    \caption{\revision{Test 3D-large: Four screenshots of the moving flock. The flock splits into two separated groups (\textbf c.).  Animated simulation available at \url{www.emilianocristiani.it/attach/starlings-3Dlarge-split.mov}} }
    \label{test:3Dlarge_revision} 
\end{figure}

%
%



\section{\revision{Discussion and biological insights}} \label{sec:conclusions} 
Theoretical and numerical results allow us to sketch some conclusions\revision{. F}irst of all, we have proved that the model is well-posed in between two changes of the set of neighbors of each bird, this guarantees that the numerical solution is meaningful. Secondly, we have isolated, in a minimal model, the features which reproduce realistic turning behavior in a setting where all agents follow the same, elementary, behavioral rules. Comparing with real birds aerial displays, our model seems to be particularly suitable for reproducing large flocks, i.e.\ flocks where two or more turnings coexist, or even the flock splits.

The \revision{model} suggests that there is no need to distinguish agents behavior on the basis of their age, gender, or position in the flock, especially between boundary and interior.  
There is also no need to assume that they even know to be on the edge of the flock, although they reasonably perceive it to some extent.
Similarly, there is no need to assume any form of ``leadership contagion'', i.e.\ the fact that an agent increases the chance to become a leader after seeing another leader. 
Of course, the fact that the model reproduces a correct behavior without contagion does not imply that contagion does not exist, at least in some form.

The presence of the delay in the dynamics seems to be very important, but removing it ($\delta=0)$ does not cause the complete loss of the desired features. 
This is perfectly in line with the results obtained in \cite{caglioti2020,cavagna2015}, where  third-order dynamics are used to correctly reproduce turns (since a second-order dynamics with delay is \revision{basically} equivalent to a third-order dynamics).
From the numerical point of view, it is important to note that we have set $\delta=\Delta t$ in the simulation, i.e.\ the delay equals the time step of the numerical approximation. 
We have observed that a larger $\delta$ does not improve the results, while a smaller $\delta$ is not acceptable in the numerical code.
Unfortunately using a smaller time step increases excessively the CPU time, especially for large flocks in 3D, for this reason we did not investigate the role of the delay further.

\revision{
Regarding the role of the parameters, we have performed a qualitative sensitive analysis for the novel key parameters regulating leaders' behavior: persistence distance $\frak d$, persistence time $\frak p$, and refractory time $\frak r$.}

\revision{Starting with the persistence distance, we observed that, tuning the value of $\frak d$ around the default value of 20 (approximately in the interval $[0,40]$), the flock moves stretching and compressing, without splitting. On the contrary, for $\frak d>40$, the flock splits into two or more groups, which could eventually merge again as a consequence of new turning phases.}

\revision{For what concerns the persistence time, for low values of $\frak p$ (approximately in the interval $[0,100]$), the flock behaves almost as in absence of leaders. As the value of $\frak p$ grows, the probability to trigger a successful turning phase increases as well.
(Note that large values of $\frak p$ are effective only if the leader is actually followed by some group mates in its turning attempt. In fact, if the leader moves away from the flock but nobody follows, it comes back to the follower status because of the other parameter $\frak d$.)} 

\revision{The sensitivity analysis for the refractory time $\frak r$ is instead more difficult, since this parameter strictly combines with the probability to become a leader (Bernoulli trial). These two parameters define, in average, the number of leaders which are active at the same time in the flock.
Numerical simulations strongly suggests that a clear, persistent change of direction of the whole flock is only possible if a critical mass of leaders is reached, i.e.\ if several leaders turn (almost) at the same time (almost) in the same direction. 
Since birds flying in the interior of the flock are subject to collisions and can change direction with difficulty, the critical mass, if any, is likely reached at the boundary of the flock. 
}

\appendix
\section{Appendix}
Here we present the proof of Proposition \ref{prop:Filippov_extension}.

 \begin{proof}
 Let us consider the intervals
 \begin{equation}
 \delta +ih \le t \le \delta + \left( i+1\right)h, \qquad \forall \ i=0,...,k-1
 \end{equation}
 where $k \ge 1$ integer, and $h=T/k$.
 We define the following approximate solution 
 
 \begin{align}\label{approx_solu}
\left\{
	\begin{array}{ll}
	y_k(t)= \phi_{0} + \displaystyle \int_{\delta}^{t} f\left(\tau, y_k(\tau-h), y_k(\tau-h-\delta) \right)\, d\tau, & \quad \delta< t <\delta+ T, \\ [3mm]
	y_k(t) \equiv \phi_{0}, & \quad t \in [-\delta,\delta].
	\end{array}
\right.
\end{align}

The functions $\{y_k\}_k$ are uniformly bounded and equicontinuous. 
In fact, from A3) we get that there exists $\varepsilon_0>0$ such that

\begin{equation}\label{unif_lim}
    \left|y_k(t)- \phi_{0} \right| \le  \displaystyle \int_{\delta}^{t} m(\tau) d\tau< \varepsilon_0.
\end{equation}

Moreover, let $\nu >0$ and $t_1 \neq t_2 $ such that $\left| t_2-t_1 \right|< \nu$.
For any $\varepsilon>0$ it holds
\begin{equation}\label{equic}
    \left|y_k(t_1)- y_k(t_2) \right| \le  \displaystyle \left|\int_{t_1}^{t_2} m(\tau) d\tau \right|  =
    \left|\varphi(t_2)-\varphi(t_1)\right|<
    \varepsilon,
\end{equation}
since $\varphi (t) \equiv \int_{\delta}^{t} m(\tau) d\tau$ is uniformly continuous.
%
%
Since \eqref{unif_lim} and \eqref{equic} hold true, Ascoli-Arzelà theorem ensures the existence of a uniformly convergent subsequence. In the following, we still denote with $\{y_k\}_k$ that subsequence, and with $y$ its limit.
By \eqref{equic}, for $h<\nu$, we get 
\begin{equation}
    \left|y_k(\tau-h)- y(\tau) \right| \le   \left|y_k(\tau-h)- y_k(\tau) \right| +
    \left|y_k(\tau)- y(\tau) \right|<
    \varepsilon,
\end{equation}
hence $y_k(\tau-h)$ converges to $y(\tau)$. In the same way, $y_k(\tau-h-\delta)$ tends to $y(\tau)$. 
Caratheodory conditions A1) and A3) allow to pass to the limit under the integral in \eqref{approx_solu}.
Hence, the limit function $y(t)$ satisfies the integral equation

\begin{equation}
y(t)=\phi_0+ \displaystyle \int_{\delta}^{t} f\left(\tau, y(\tau-h), y(\tau-h-\delta) \right)\, d\tau,
\end{equation}
which means that $y(t)$ is a solution to \eqref{model_delay}.


We now prove the uniqueness of the solution to (\ref{model_delay}).
Let $y_1, y_2 : (0, \delta+T) \rightarrow \Omega $ be solutions to problem (\ref{model_delay}), and consider $z(t)= y_1(t)-y_2(t) $. By A4) we get that for every $t \in (0,\delta+T)$,  
\begin{equation}
\left| z(t) \right| \le \displaystyle \int_{\delta}^{t} l(\tau)\left( \left| z(\tau)\right| + \left| z(\tau-\delta)\right| \right) d\tau
\le 2 \displaystyle \int_{\delta}^{t} l(\tau) \sup_{0\le s \le \tau}  \left| z(s)\right| d\tau.
\end{equation}
It follows that
\begin{equation}
\sup_{0\le s \le t}  \left| z(s)\right|  \le 2 \displaystyle \int_{\delta}^{t} l(\tau) \sup_{0\le s \le \tau}  \left| z(s)\right| d\tau.
\end{equation}
We conclude that $\displaystyle \sup_{0\le s \le t}  \left| z(s)\right| \equiv 0$ for any $t \in (\delta, \delta+T)$, hence $y_1 \equiv y_2$. 
\qed
\end{proof}

\noindent
\revision{\textbf{Authors' contribution}
EC conceived the idea, developed the model, proved the theoretical results, wrote the numerical code, wrote the manuscript, and got fundings for the research.
MM developed the model, proved the theoretical results, wrote the numerical code, carried out the numerical experiments, and wrote the manuscript.
MP proved the theoretical results.
LB developed the model.}

\begin{acknowledgements}
The authors want to warmly thank Andrea Cavagna, Irene Giardina, and Giovanna Nappo for the useful discussions we had during the preparation of the manuscript. 
\end{acknowledgements}

\noindent
\textbf{Conflict of interest}
The authors declare that they have no conflict of interest.

\bibliographystyle{plain}
\bibliography{biblio}

\end{document}